\pdfoutput=1
\RequirePackage{ifpdf}
\ifpdf 
\documentclass[pdftex]{sigma}
\else
\documentclass{sigma}
\fi

\usepackage{bbm}
\usepackage[all]{xy}
\usepackage{upgreek}

\numberwithin{equation}{section}

\newtheorem{theom}{Theorem}[section]
\newtheorem{coroll}[theom]{Corollary}
\newtheorem{lema}[theom]{Lemma}
\newtheorem{propo}[theom]{Proposition}

\theoremstyle{definition}
\newtheorem{defin}[theom]{Definition}
\newtheorem{Remark}[theom]{Remark}

\newcommand{\C}{\mathbb{C}}
\newcommand{\R}{\mathbb{R}}
\newcommand{\ebar}{\overline{\varepsilon}}

\begin{document}

\allowdisplaybreaks

\newcommand{\arXivNumber}{1906.07926}

\renewcommand{\PaperNumber}{098}

\FirstPageHeading

\ShortArticleName{Exact Bohr--Sommerfeld Conditions for Quantum Periodic Benjamin--Ono}

\ArticleName{Exact Bohr--Sommerfeld Conditions\\ for the Quantum Periodic Benjamin--Ono Equation}

\Author{Alexander MOLL}

\AuthorNameForHeading{A.~Moll}

\Address{Department of Mathematics, Northeastern University, Boston, MA~USA}
\Email{\href{mailto:a.moll@northeastern.edu}{a.moll@northeastern.edu}}
\URLaddress{\url{https://web.northeastern.edu/moll/}}

\ArticleDates{Received June 20, 2019, in final form December 12, 2019; Published online December 18, 2019}

\Abstract{In this paper we describe the spectrum of the quantum periodic Benjamin--Ono equation in terms of the multi-phase solutions of the underlying classical system (the periodic multi-solitons). To do so, we show that the semi-classical quantization of this system given by Abanov--Wiegmann is exact and equivalent to the geometric quantization by Nazarov--Sklyanin. First, for the Liouville integrable subsystems defined from the multi-phase solutions, we use a result of G\'erard--Kappeler to prove that if one neglects the infinitely-many transverse directions in phase space, the regular Bohr--Sommerfeld conditions on the actions are equivalent to the condition that the singularities of the Dobrokhotov--Krichever multi-phase spectral curves define an anisotropic partition (Young diagram). Next, we locate the renormalization of the classical dispersion coefficient by Abanov--Wiegmann in the realization of Jack functions as quantum periodic Benjamin--Ono stationary states. Finally, we show that the classical energies of Bohr--Sommerfeld multi-phase solutions in the renormalized theory give the exact quantum spectrum found by Nazarov--Sklyanin without any Maslov index correction.}


\Classification{37K40; 37K10; 53D50; 81Q20; 81Q80}

\section{Introduction and statement of result} \label{SECIntro}

\looseness=1 In the semi-classical analysis of quantized Hamiltonian systems, a major goal is to approximate the quantum spectrum in terms of select periodic orbits of the underlying classical system. In special cases, a semi-classical approximation of the quantum spectrum may turn out to be \textit{exact}. For quantizations of Liouville integrable systems, the classical energies of orbits satisfying the regular Bohr--Sommerfeld conditions give an approximation to the spectrum which is exact, e.g., for free particles on tori. Similarly, the WKB matching conditions provide an approximate spectrum which is exact, e.g., for harmonic oscillators. For quantized chaotic systems, the semi-classical approximation in the Gutzwiller trace formula is also exact in several cases, e.g., for free particles on any surface of constant negative curvature. For background on semi-classical and geometric quantization, see Kirillov~\cite{Kirillov1990} and Takh\-tajan~\cite{TakhtajanBOOK}.

In this paper we give an exact semi-classical description of the spectrum of the quantum~perio\-dic Benjamin--Ono equation in terms of distinguished quasi-periodic orbits of the underlying classical system known as \textit{multi-phase solutions}, the periodic analogs of multi-soliton solutions. We state this result in Theorem~\ref{MAINTHEOREM} below. For recent mathematical surveys of classical and quantum periodic Benjamin--Ono equations, see \cite[Section~5.2]{Saut2018} and~\cite[Section~1.1.6]{Okounkov2018ICM}, re\-spec\-ti\-ve\-ly.

\subsection{1-phase solutions of classical Benjamin--Ono} Let $J$ be the spatial Hilbert transform
defined by
\begin{gather*} (J\varphi) (x) = \text{P.V.} \frac{1}{\pi} \int_{- \infty}^{+\infty} \frac{\varphi(y){\rm d}y}{ x-y}\end{gather*} with $J {\rm e}^{{\bf i}kx} = -{\bf i} \operatorname{sgn}(k){\rm e}^{{\bf i}kx}$. The Benjamin--Ono equation \cite{Benj, DavisAcrivos, Ono}
\begin{gather} \label{CBOE} \partial_t v + v \partial_x v = \tfrac{1}{2} \ebar J\big[ \partial^2_{x} v\big] \end{gather} for real $v(x,t)$, $x ,t \in \R$, of spatial period $2\pi$ is Hamiltonian for the Gardner--Faddeev--Zakharov bracket as we review in Section~\ref{SECclassicalBO}. In~(\ref{CBOE}), $\ebar > 0$ is a coefficient of dispersion whose notation we explain in Section~\ref{SECnekrasov}. For $\mathbb{T} = \R / 2 \pi \mathbb{Z}$, Molinet \cite{Molinet} proved~(\ref{CBOE}) is globally well-posed in~$L^2(\mathbb{T})$. We write $v(x,t; \ebar)$ for solutions of~(\ref{CBOE}). In their original derivation and analysis of the equation~(\ref{CBOE}), both Benjamin \cite{Benj} and Ono~\cite{Ono} found a 3-parameter family of periodic traveling waves that define periodic orbits of~(\ref{CBOE}) known as \textit{$1$-phase solutions}:
 \begin{defin} \label{DEF1phase} For any $3$ real parameters $\vec{s} \in \R^{3}$ ordered as \begin{gather} \label{BO1PhaseParameters}
 s_1^{\uparrow} < s_1^{\downarrow} < s_0^{\uparrow} \end{gather}
 and $1$ phase $\chi_1 \in \mathbb{R}$, the 1-phase solutions of (\ref{CBOE}) are the periodic traveling waves
\begin{gather} \label{BO1PhaseSolution} v^{\vec{s}, \chi_1} (x,t; \ebar) = \eta^{\vec{s}}(x - \chi_1 - c_1(\vec{s}) t; \ebar) \end{gather}
with wavespeed \begin{gather} \label{BO1PhaseWaveSpeed} c_1(\vec{s}) = \tfrac{1}{2} \big( s_1^{\downarrow} + s_0^{\uparrow}\big) \end{gather} and permanent form \begin{gather} \label{BO1PhaseForm} \eta^{\vec{s}}(x; \ebar) = \frac{ \big(s_0^{\uparrow} - s_1^{\downarrow}\big)^2}{ \big(s_1^{\downarrow} - s_1^{\uparrow}\big) + \big(s_0^{\uparrow} - s_1^{\uparrow}\big) - 2 \sqrt{ \frac{ s_0^{\uparrow} - s_1^{\uparrow}}{ s_1^{\downarrow} - s_1^{\uparrow}}} \cos \big ( \big ( \frac{ s_0^{\uparrow} - s_1^{\downarrow}}{\ebar} \big ) x\big )}. \end{gather}
 \end{defin}
 The parameters (\ref{BO1PhaseParameters}) define two closed intervals $\big({-} \infty, s_1^{\uparrow}\big]$, $\big[s_1^{\downarrow}, s_0^{\uparrow}\big]$ we call \textit{bands} and one open interval $\big(s_1^{\uparrow}, s_1^{\downarrow}\big)$ we call a \textit{gap}. The wavespeed (\ref{BO1PhaseWaveSpeed}) is the midpoint of the band $\big[s_1^{\downarrow}, s_0^{\uparrow}\big]$. The wavelength of (\ref{BO1PhaseForm}) is inversely proportional to the length of the band $\big[s_1^{\downarrow}, s_0^{\uparrow}\big]$ and proportional to $\ebar$. As the band $\big[s_1^{\downarrow}, s_0^{\uparrow}\big]$ shrinks or merges with the band $\big({-} \infty, s_1^{\uparrow}\big]$, the permanent form (\ref{BO1PhaseForm})
 \begin{gather*} \eta^{\vec{s}} (x; \ebar) \sim \begin{cases} s_1^{\uparrow} + \frac{ \ebar^2 }{ \frac{1}{2 \big(s_0^{\uparrow} - s_1^{\uparrow}\big)^2} + x^2},& s_1^{\uparrow} < s_1^{\downarrow} \longrightarrow s_0^{\uparrow}, \\ s_0^{\uparrow} , & s_1^{\uparrow} \longleftarrow s_1^{\downarrow} < s_0^{\uparrow} \end{cases}
 \end{gather*}
of the $1$-phase solution (\ref{BO1PhaseSolution}) converges to that of a $1$-soliton or constant solution of (\ref{CBOE}).

\subsection{Multi-phase solutions of classical Benjamin--Ono}
After the discovery of the family of 1-phase periodic orbits (\ref{BO1PhaseSolution}), Satsuma--Ishimori \cite{SatsumaIshimori1979} found a larger family of quasi-periodic orbits of (\ref{CBOE}) known as \textit{multi-phase solutions}. We now recall a formula for these multi-phase solutions due to Dobrokhotov--Krichever \cite{DobrokhotovKrichever}. Throughout we write $\delta$ for Kronecker delta.
\begin{defin} \label{DEFmultiphase} For $n=0,1,2,\ldots$ with $2n+1$ real parameters $\vec{s} \in \R^{2n+1}$ ordered as \begin{gather} \label{BOMultiPhaseParameters} s_n^{\uparrow} < s_n^{\downarrow} < \cdots < s_1^{\uparrow} < s_1^{\downarrow} < s_0^{\uparrow} \end{gather} and $n$ phases $\chi_n, \ldots, \chi_1 \in \mathbb{R}$ denoted $\vec{\chi} \in \mathbb{R}^n$, the multi-phase ($n$-phase) solutions of (\ref{CBOE}) are
\begin{gather} \label{BOmPhaseSolution}
v^{\vec{s}, \vec{\chi}} (x, t; \ebar) = s_n^{\uparrow} - \sum_{i=1}^n \big(s_{i-1}^{\uparrow} - s_i^{\downarrow}\big) - 2 \ebar \operatorname{Im} \partial_x \log \det M^{\vec{s}, \vec{\chi}} (x,t; \ebar ),
\end{gather} where $M^{\vec{s}, \vec{\chi}}(x,t;\ebar)$ is the $n \times n$ matrix with entries $M^{\vec{s}, \vec{\chi}}_{ij} (x,t; \ebar)$ for $1 \leq i,j \leq n$ defined by
\begin{gather} \label{BOmPhaseSolutionMatrix}
 M^{\vec{s}, \vec{\chi}}_{ij} (x,t; \ebar)=\tfrac{1}{ s_{i-1}^{\uparrow} - s_j^{\downarrow}} \Big ( {-} 1 + \delta(i-j) Z_i(\vec{s}) {\rm e}^{ {{\bf i} \big( \tfrac{s_{i-1}^{\uparrow} - s_i^{\downarrow} }{\ebar}\big) \big( x - \chi_i - \tfrac{1}{2}\big(s_i^{\downarrow} + s_{i-1}^{\uparrow}\big)t \big) } } \Big ),\\
 Z_i(\vec{s}) = \sqrt{ \tfrac{ s_{i-1}^{\uparrow} - s_n^{\uparrow} }{ s_{i}^{\downarrow} - s_n^{\uparrow}} } \prod_{j \neq i} \sqrt{ \tfrac{ \big(s_i^{\downarrow} - s_j^{\downarrow}\big)\big( s_{i-1}^{\uparrow} - s_{j-1}^{\uparrow}\big)}{ \big(s_{i-1}^{\uparrow} - s_j^{\downarrow}\big) \big( s_i^{\downarrow} - s_{j-1}^{\uparrow}\big)}} .\end{gather}
\end{defin}

For $n=1$, (\ref{BOmPhaseSolution}) is the 1-phase periodic traveling wave (\ref{BO1PhaseSolution}). In Dobrokhotov--Kriche\-ver~\cite{DobrokhotovKrichever}, the parameters (\ref{BOMultiPhaseParameters}) are singularities of their rational spectral curves. As in the $1$-phase case, (\ref{BOMultiPhaseParameters}) defines $n+1$ closed intervals $\big({-} \infty, s_n^{\uparrow}\big]$, $\big[s_n^{\downarrow}, s_{n-1}^{\uparrow}\big], \ldots, \big[s_1^{\downarrow}, s_0^{\uparrow}\big]$ we call \textit{bands} and $n$ open intervals $\big(s_n^{\uparrow}, s_n^{\downarrow}\big), \ldots, \big(s_1^{\uparrow}, s_1^{\downarrow}\big)$ we call \textit{gaps}. As all band lengths shrink $ | s_i^{\downarrow} - s_{i-1}^{\uparrow}| \rightarrow 0$, the multi-phase solution~(\ref{BOmPhaseSolution}) converges to one of the multi-soliton solutions of~(\ref{CBOE}) found by Matsuno~\cite{Matsuno1979}.

\subsection{Bands and spatial periodicity conditions} The form of exponential terms in (\ref{BOmPhaseSolutionMatrix}) implies:

\begin{propo} \label{MultiPhasePeriodicityConditions} The $n$-phase solution $v^{\vec{s}, \vec{\chi}}(x,t; \ebar)$ in \eqref{BOmPhaseSolution} is $2\pi$-periodic in $x$ if and only if
\begin{gather} \label{BOMultiPhasePeriodicConditions} \big| s_i^{\downarrow} - s_{i-1}^{\uparrow} \big| = \ebar N_i \end{gather} the length of the $i$th band $[s_i^{\downarrow}, s_{i-1}^{\uparrow}]$ is a positive integer $N_i \in \mathbb{Z}_+$ multiple of $\ebar$ for all $i=1,\ldots, n$, i.e., the $i$th $1$-phase periodic traveling wave in the $n$-phase wave has $N_i$ bumps on $\mathbb{T} \cong \R / 2 \pi \mathbb{Z}$.\end{propo}

The $n=1$ case of Proposition \ref{MultiPhasePeriodicityConditions} follows also by direct inspection of the cosine term in~(\ref{BO1PhaseForm}). We now show that additional conditions on lengths $\big| s_i^{\uparrow} - s_i^{\downarrow}\big|$ of gaps $\big(s_i^{\uparrow}, s_{i}^{\downarrow}\big)$ arise in quantization.

\subsection{Statement of result: gaps and exact Bohr--Sommerfeld conditions}\label{SECstatement}
In Theorem~\ref{MAINTHEOREM} below, we give exact Bohr--Sommerfeld quantization conditions on the tori in phase space defined by the classical multi-phase solutions (\ref{BOmPhaseSolution}). As a~consequence, we give a~semi-classical interpretation of the results of Nazarov--Sklyanin~\cite{NaSk2}
for the quantum periodic Benjamin--Ono equation and also show that the semi-classical soliton quantization of~(\ref{CBOE}) by Abanov--Wiegmann~\cite{AbWi1} is exact. Recall for $\hbar >0$ and any periodic orbit $\upgamma$ of a classical Hamiltonian $O\colon M \rightarrow \R$ in a phase space $(M, d \upalpha)$ with Liouville $1$-form $\upalpha$, the \textit{$\hbar$-Bohr--Sommerfeld condition} on $\upgamma$ is that the action $\oint_{\upgamma} \upalpha$ is a $N' \in \mathbb{Z}_+$ multiple of $2 \pi \hbar$:
 \begin{gather} \label{BSConditionIntro} \oint_{\upgamma} \upalpha = 2 \pi \hbar N' .\end{gather}
 Recall also that for a self-adjoint operator $\widehat{O}(\hbar)$ in a Hilbert space $(\mathcal{H}, \langle
 \cdot, \cdot \rangle)$ chosen to quantize~$O$ in $(M, d \upalpha)$, the \textit{$\hbar$-Bohr--Sommerfeld
 approximation} to the spectrum of $\widehat{O}(\hbar)$ is given by the classical energies $O|_{\upgamma}$ of the
 periodic orbits $\upgamma$ satisfying~(\ref{BSConditionIntro}). When~$O$ is Liouville integrable, one takes
 conditions (\ref{BSConditionIntro}) for each Hamiltonian $O_i$ in an integrable hierarchy containing $O$ whose
 corresponding periodic orbits $\upgamma_i$ are a basis of cycles on the Liouville tori. Using a~recent description
 of (\ref{CBOE}) as a classical Liouville integrable Hamiltonian system in~$L^2(\mathbb{T})$ by G\'erard--Kappeler
 \cite{GerardKappeler2019}, in Section \ref{SECquantumMP} we prove:

\begin{theom}\label{MAINTHEOREM} Let $\ebar >0$ and $\hbar>0$ be dimensionless coefficients of dispersion and quantization.
\begin{itemize}\itemsep=0pt
\item {\rm [Part I: Bohr--Sommerfeld conditions]} Let $\upgamma_{i,n}^{\vec{s}}(\ebar)$ be the cycle in the phase space of \eqref{CBOE} defined by varying only the $i$th phase $\chi_i$ in the multi-phase initial data $v^{\vec{s}, \vec{\chi}}(x, 0; \ebar)$ \eqref{BOmPhaseSolution}. The action of the Gardner--Faddeev--Zakharov Liouville $1$-form $\upalpha_{\rm GFZ}$ \eqref{1formFORMULA} along $\upgamma_{i,n}^{\vec{s}} (\ebar)$ is \begin{gather} \label{ClassicalActionsRESULT} \oint_{\upgamma_{i,n}^{\vec{s}} (\ebar)} \upalpha_{\rm GFZ} = 2 \pi \ebar \big|s_i ^{\uparrow} - s_i^{\downarrow} \big| \end{gather} $2\pi \ebar$ times the length of the $i$th gap $\big(s_i^{\uparrow}, s_{i}^{\downarrow}\big)$. Neglecting the infinitely-many transverse directions in phase space to $n$-phase tori, the regular Bohr--Sommerfeld conditions~\eqref{BSConditionIntro} on the classical actions~\eqref{ClassicalActionsRESULT} are therefore \begin{gather} \label{BOMultiPhaseBohrSommerfeldConditions}\big|s_i ^{\uparrow} - s_i^{\downarrow} \big| = \frac{ \hbar}{\ebar} N_i' \end{gather} that the length of the $i$th gap $\big(s_i^{\uparrow}, s_i^{\downarrow}\big)$ is a positive integer $N_i' \in \mathbb{Z}_+$ multiple of $\hbar / \ebar >0$ for all $i=1,\ldots, n$ with $N_i'$ independent of $N_i$ in the description of band lengths in \eqref{BOMultiPhasePeriodicConditions}.
\item {\rm [Part II: Exact Bohr--Sommerfeld conditions]} The spectrum of the geometric quantization of \eqref{CBOE} on $\mathbb{T}$ for fixed $a= \int_0^{2\pi} v(x) \tfrac{{\rm d}x}{2 \pi}$ found by Nazarov--Sklyanin~{\rm \cite{NaSk2}} is the subset of classical energy
 levels $\sum\limits_{i=0}^n \big(s_i^{\uparrow}\big)^3 - \sum\limits_{i=1}^n \big(s_i^{\downarrow}\big)^3$ of the $n$-phase solutions \eqref{BOmPhaseSolution} for $n=0,1,2,\ldots$ given by $\vec{s} \in \R^{2n+1}$ satisfying inequalities \eqref{BOMultiPhaseParameters},
 $ a = \sum\limits_{i=0}^n s_i^{\uparrow} - \sum\limits_{i=1}^n s_i^{\downarrow} $, the spatial periodicity conditions \eqref{BOMultiPhasePeriodicConditions}, and the Bohr--Sommerfeld conditions \eqref{BOMultiPhaseBohrSommerfeldConditions} with $\ebar$ in each replaced by \begin{gather} \label{Renormalization} \varepsilon_1 (\ebar, \hbar) =\frac{ \ebar + \sqrt{\ebar^2 + 4 \hbar}}{ 2} \end{gather} the renormalized coefficient of classical dispersion determined by Abanov--Wiegmann~{\rm \cite{AbWi1}}.\end{itemize}
\end{theom}

\begin{figure}[htb]\centering
\includegraphics[width=0.25 \textwidth]{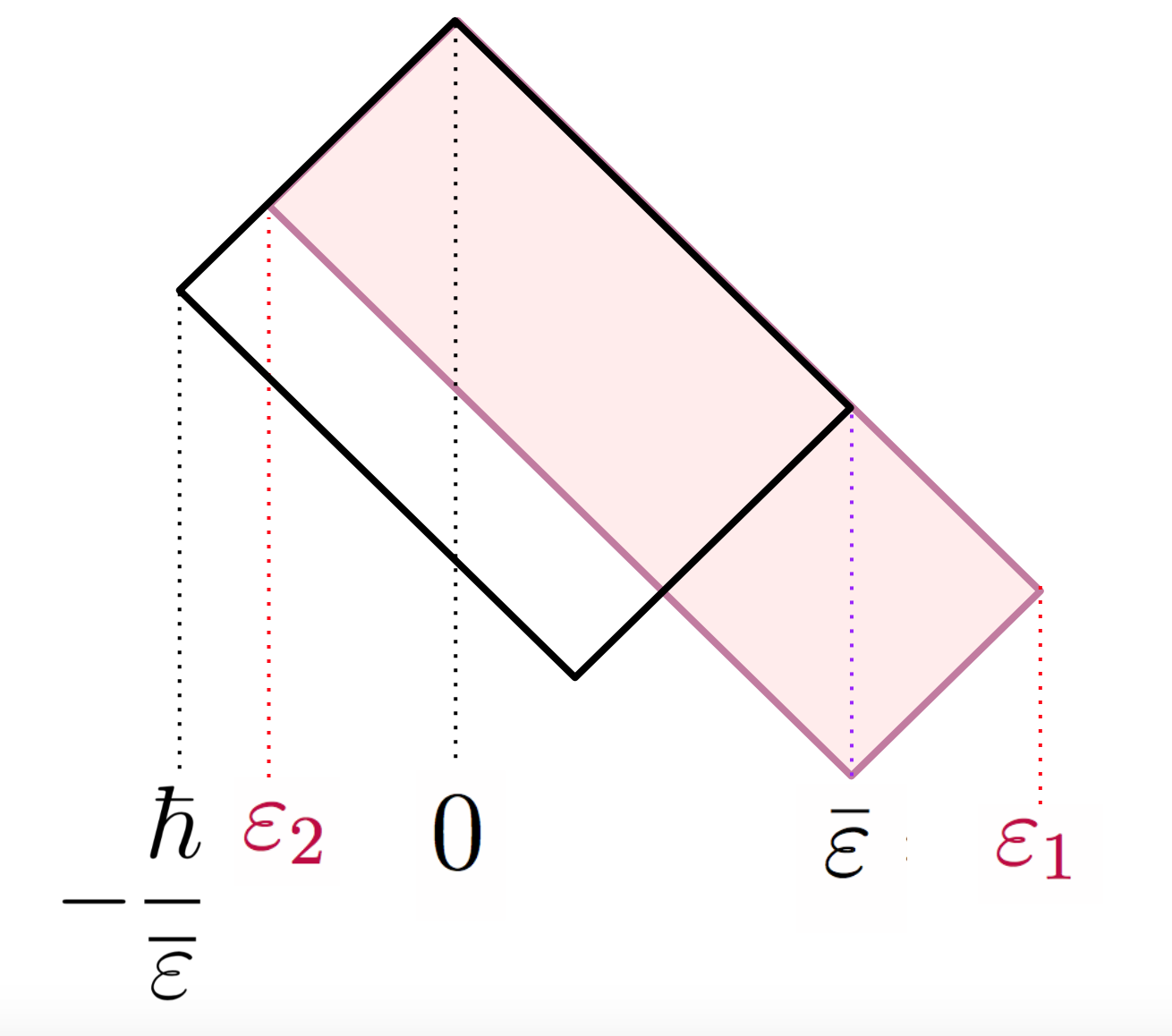}
\caption{Renormalization.}\label{FIGrectangle}
\end{figure}

The renormalization (\ref{Renormalization}) in Part~II of Theorem~\ref{MAINTHEOREM} reflects the fluctuations of the quantum system in the infinitely-many transverse
directions to the Liouville tori neglected in Part~I. The formula~(\ref{Renormalization}) for $\varepsilon_1$ can be characterized by the rectangles $R(r_2, r_1)$ of side lengths $- r_2 \sqrt{2}, r_1 \sqrt{2}$ for two choices of $r_2 < 0 <r_1$ in Fig.~\ref{FIGrectangle}: for
\begin{gather}\label{RenormalizationCompanion}
 \varepsilon_2 ( \ebar, \hbar) = \frac{ \ebar - \sqrt{\ebar^2 + 4 \hbar}}{2},
\end{gather}
 (i) $R({\varepsilon_2, \varepsilon_1})$ encloses the same area $2 \hbar$ as $R(-{\hbar / \ebar, \ebar})$ and (ii) the intersection of $R({\varepsilon_2, \varepsilon_1})$ and the exterior of $R(-{\hbar / \ebar, \ebar})$ is a square.

\subsection[Interpretation of result: Dobrokhotov--Krichever profiles and anisotropic partitions]{Interpretation of result: Dobrokhotov--Krichever profiles\\ and anisotropic partitions} \label{SUBSECinterpretation} We now interpret our Theorem~\ref{MAINTHEOREM} for the quantum multi-phase solutions in terms of partitions (Young diagrams) built from the rectangles in Fig.~\ref{FIGrectangle}. For the classical multi-phase solutions,

\begin{defin} \label{MultiPhaseProfileDEF} A Dobrokhotov--Krichever profile $f(c |\vec{s})$ is a piecewise-linear function of $c \in \R$ with slopes $\pm 1$, local extrema (\ref{BOMultiPhaseParameters}), and $f(c | \vec{s}) \sim |c- a|$ as $c \rightarrow \pm \infty$ for $a = \sum\limits_{i=0}^n s_i^{\uparrow} - \sum\limits_{i=1}^n s_i^{\downarrow}$.
\end{defin}

A subset of Dobrokhotov--Krichever profiles are the anisotropic partition profiles of Kerov~\cite{Ke4}:
\begin{defin} \label{AnisotropicProfileDEF} For $r_2< 0 < r_1$ and $a \in \R$, a piecewise-linear real function $f(c)$ of $c \in \R$ is an {anisotropic partition profile of anisotropy $(r_2, r_1)$ centered at $a$} if the region
\begin{gather} \label{TheRegion} \big\{ (c,y) \in \R^2 \colon |c - a | < y < f(c) \big\} \end{gather} is a disjoint union of finitely-many translates of a $-r_2 \sqrt{2} \times r_1 \sqrt{2}$ rectangle $R(r_2, r_1)$.
 \end{defin}

\begin{figure}[htb]\centering
\includegraphics[width=0.45\textwidth]{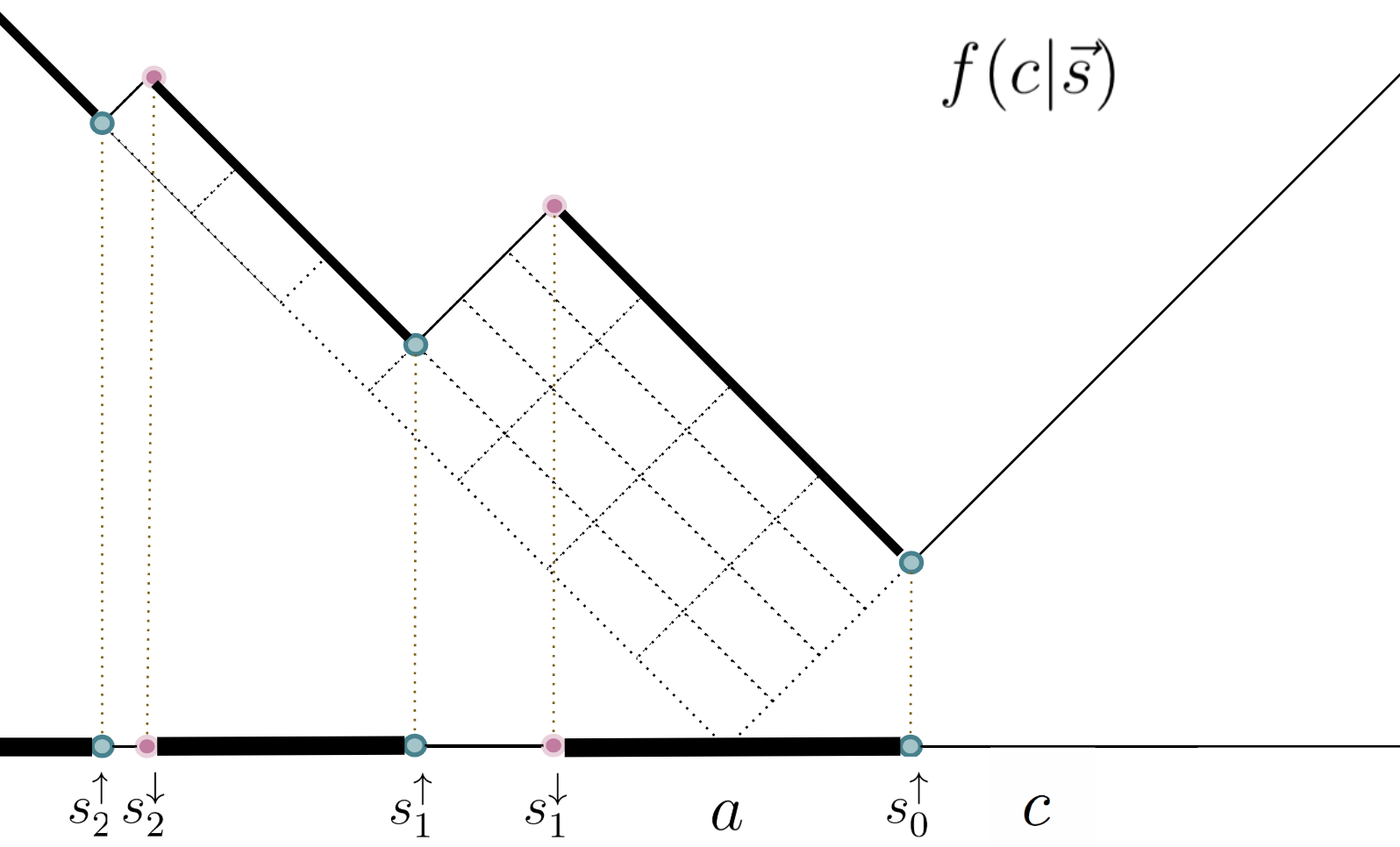}
\caption{Dobrokhotov--Krichever profile $f( c | \vec{s})$ of quantum Benjamin--Ono $2$-phase solution with gap and band lengths determined by $(N_2', N_2, N_1', N_1) = (1, 3, 3, 4)$.}\label{anisotropicPOPFIGURE}
\end{figure}

\begin{propo}\label{INTERPRETATION}Theorem~{\rm \ref{MAINTHEOREM}} can be restated for Dobrokhotov--Krichever profiles $f( c| \vec{s})$:
\begin{itemize}\itemsep=0pt
\item {\rm [Part I]} The original $($approximate$)$ regular Bohr--Sommerfeld conditions on the multi-phase $v^{\vec{s}, \vec{\chi}}(x, t; \ebar)$ are that $f(c | \vec{s})$ is an anisotropic partition profile of anisotropy $(-\hbar / \ebar, \ebar)$.
\item {\rm [Part II]} The renormalized $($exact$)$ regular Bohr--Sommerfeld conditions on the multi-phase $v^{\vec{s}, \vec{\chi}}(x, t; \ebar)$ are that $f(c | \vec{s})$ is an anisotropic partition profile of anisotropy $(\varepsilon_2, \varepsilon_1)$.
\end{itemize}
\end{propo}

\begin{proof} $f(c |\vec{s})$ is an anisotropic partition profile of anisotropy $(r_2, r_1)$ if and only if the band lengths $\big|s_i^{\downarrow} - s_{i-1}^{\uparrow} \big| = r_1 N_i$ and gap lengths $\big|s_i^{\uparrow} - s_i^{\downarrow} \big| = -r_2 N_i ' $ for $N_i, N_i ' \in \mathbb{Z}_+$ for $i=1, \ldots, n$.\end{proof}

\subsection{Motivation} \label{SECmotivation} Our Theorem~\ref{MAINTHEOREM} establishes the presence of classical multi-phase solutions in the work of Nazarov--Sklyanin \cite{NaSk2} and hence realizes Jack functions as quantum multi-phase states. Conversely, our Theorem~\ref{MAINTHEOREM} relates the results in Nazarov--Sklyanin \cite{NaSk2} to the semi-classical studies of quantum Benjamin--Ono dynamics out of equilibrium by Abanov--Wiegmann \cite{AbWi1}, Bettelheim--Abanov--Wiegmann~\cite{AbBeWi2}, and Wiegmann~\cite{Wieg1}. As will appear in~\cite{Moll5}, Theorem~\ref{MAINTHEOREM} implies that the semi-classical and small dispersion asymptotics in the author's thesis~\cite{Moll0} on Jack measures, a generalization of Okounkov's Schur measures~\cite{Ok1}, reflect the structure of quantum dispersive shock waves and quantum soliton trains emitted by coherent states as studied in Bettelheim--Abanov--Wiegmann~\cite{AbBeWi2}. Note that a classical version of these small dispersion asymptotics, relating dispersive action profiles of the classical hierarchy in~\cite{NaSk2} to the formation of classical dispersive shock waves, have already been discussed by the author in \cite[Section~8]{Moll1}.

\subsection{Outline} \label{SECoutline} In Section~\ref{SECcomments} we discuss our Theorem~\ref{MAINTHEOREM} and its relation to previous results. In Section~\ref{SECclassicalBO} we review the Hamiltonian and Lax operator of~(\ref{CBOE}). In Section~\ref{SECclassicalNS} we recall the classical Nazarov--Sklyanin hierarchy \cite{NaSk2} and its presentation in terms of dispersive action profiles from \cite{Moll1}. In Section~\ref{SECclassicalINT} we identify the classical global action variables of G\'erard--Kappeler~\cite{GerardKappeler2019} with the gaps in the dispersive action profiles from \cite{Moll1}. In Section~\ref{SECquantumBO} we define a geometric quantization of~(\ref{CBOE}) by quantizing the Hamiltonian and Lax operator. In Section~\ref{SECquantumINT} we show that the renormalization of the classical coupling (\ref{Renormalization}) in Abanov--Wiegmann~\cite{AbWi1} is implicit in the realization of Jack functions as quantum periodic Benjamin--Ono Hamiltonian eigenfunctions. In Section~\ref{SECquantumNS} we present the quantum Nazarov--Sklyanin hierarchy and its exact spectrum from~\cite{NaSk2}. In Section~\ref{SECclassicalMP} we recall finite gap conditions for multi-phase solutions from \cite{Moll1}. In Section~\ref{SECquantumMP} we derive formula (\ref{ClassicalActionsRESULT}) for the classical actions from results of G\'erard--Kappeler \cite{GerardKappeler2019}, establish the Bohr--Sommerfeld conditions (\ref{BOMultiPhaseBohrSommerfeldConditions}), and prove Theorem~\ref{MAINTHEOREM}.

\section{Comments and comparison with previous results} \label{SECcomments}

\subsection{Comparison with results for the sine-Gordon equation}
The \textit{WKB matching conditions} $\oint_{\upgamma} \big( \upalpha + \tfrac{\upmu_{\upgamma}}{4} \big) = 2 \pi \hbar N'$ are often said to ``correct'' (\ref{BSConditionIntro}) by the Maslov index $\oint_{\upgamma} \upmu_{\upgamma}$. However, in our Theorem~\ref{MAINTHEOREM}, the Bohr--Sommerfeld conditions are exact after renormalization without need of the Maslov index correction. This phenomena also occurs in the quantum sine-Gordon equation as emphasized in surveys by Coleman \cite{ColemanLumps}, Faddeev \cite{Faddeev40Years}, and Sklyanin--Smirnov--Takhtajan in \cite{FaddeevStudentSURVEY2017}. The renormalization (\ref{Renormalization}) for Benjamin--Ono in Abanov--Wiegmann \cite{AbWi1} is an analog of that in the semi-classical quantization of the sine-Gordon equation by Dashen--Hasslacher--Neveu \cite{DHN1974a, DHN1974b, DHN1975a}, Goldstone--Jackiw \cite{GoldstoneJackiw1975}, and Faddeev--Korepin \cite{KorepinFaddeev1975}. Our Theorem~\ref{MAINTHEOREM} is an analog of the result in Faddeev--Sklyanin--Takhatajan \cite{FaddeevSklyaninTakhtajan1979} and Coleman~\cite{Coleman1975} that this semi-classical quantization is exact.

\subsection{Comparison with results for the Calogero--Sutherland equation}

In \cite{AbWi1}, Abanov--Wiegmann give two derivations of the renormalization (\ref{Renormalization}). The first derivation in field theory is at 1-loop by an effective action and choice of counterterms in a semi-classical quantization of (\ref{CBOE}) following Jevicki \cite{Jevicki1992}. Part I of our Theorem~\ref{MAINTHEOREM} is a Hamiltonian counterpart to the 0-loop step in this first derivation, neglecting the infinitely-many transverse directions in the phase space of classical fields. The second derivation of~(\ref{Renormalization}) in hydrodynamics in \cite{AbWi1} uses the realization of (\ref{CBOE}) in Calogero--Sutherland hydrodynamics and builds upon the work of Andri\'{c}--Bardek \cite{AndricBardek1988}, Polychronakos \cite{Poly1995}, and Awata--Matsuo--Odake--Shiraishi \cite{AwMtOdSh}. For the quantum Calogero--Sutherland many-body problem, the analog of Part~II of our Theorem~\ref{MAINTHEOREM} -- namely, that after a shift the semi-classical quantization is exact -- is well-known: see reviews by Calogero~\cite{Calogero2008}, Etingof~\cite{Eti0}, Ruijenaars~\cite{Ruijsenaars1999}, and Sutherland~\cite{Suth0}. The works of Nazarov--Sklyanin~\cite{NaSk2, NaSk1} and Sergeev--Veselov \cite{SergVes, SergVes2} are exact extensions of the second hydrodynamic derivation of~(\ref{Renormalization}) in~\cite{AbWi1}.

\subsection{Comments on Hilbert schemes of points on surfaces} \label{SECnekrasov} Our notation $\ebar$ in (\ref{CBOE}), $\varepsilon_1$ in (\ref{Renormalization}), $\varepsilon_2$ in (\ref{RenormalizationCompanion}), and $a $ in (\ref{ClassicalPhaseSpace}) reflect the appearance of the quantum periodic Benjamin--Ono equation in equivariant cohomology of Hilbert schemes of points in $\C^2$ reviewed in \cite[Section~1.1.6]{Okounkov2018ICM}. To interpret our Theorem~\ref{MAINTHEOREM} in this context, note that our coefficient of classical dispersion $\ebar = \varepsilon_1+ \varepsilon_2$ is the deformation parameter of the Maulik--Okounkov Yangian \cite{MaulOk} while our coefficient of quantization $\hbar = - \varepsilon_1 \varepsilon_2$ is the handle-gluing element in \cite{MaulOk}. These $\ebar$, $\hbar$ appear in \cite[Section~1.1.2]{Okounkov2018ICM} in trading $\C^2$ for a~surface~$\mathcal{S}$. For other exact Bohr--Sommerfeld conditions in the related theory of Nekrasov \cite{Nek1}, see Mironov--Morozov~\cite{MironovMorozov2010}. Note that we have related dispersive action profiles of~(\ref{CBOE}) to profiles in Nekrasov--Shatashvili \cite{NekShat} and Nekrasov--Pestun--Shatashvili \cite{NekPesSha} in~\cite[Section~2.5]{Moll1}.

\section[Classical periodic Benjamin--Ono: Hamiltonian and Lax operator]{Classical periodic Benjamin--Ono:\\ Hamiltonian and Lax operator} \label{SECclassicalBO}

In this section we recall the formulation of the classical Benjamin--Ono equation (\ref{CBOE}) for $v$ periodic in~$x$ as a classical Hamiltonian system with respect to the Gardner--Faddeev--Zakharov symplectic form $\upomega_{\rm GFZ}$ and discuss a complex structure~$J$ on the classical phase space. We also introduce the classical Lax operator $L_{\bullet}(v; \ebar)$ of (\ref{CBOE}) and express the classical Hamiltonian in terms of $L_{\bullet}(v; \ebar)$.

\subsection[Classical phase space as Sobolev space from Gardner--Faddeev--Zakharov construction]{Classical phase space as Sobolev space\\ from Gardner--Faddeev--Zakharov construction} \label{SECclassicalPhaseSpace}

Define Fourier coefficients of $v\colon \mathbb{T} \rightarrow \R$ with the sign convention $ V_k = \int_0^{2 \pi} {\rm e}^{{\bf i} k x} v(x) \frac{ {\rm d}x}{2\pi}$ so that
\begin{gather} \label{FourierDEF} v(x) = \sum_{k \in \mathbb{Z}} V_k {\rm e}^{- {\bf i} kx}. \end{gather}
For real-valued $v$, $V_{-k} = \overline{V_k}$ for all $k \in \mathbb{Z}$.
For $a \in \R$, we choose as the classical phase space of~(\ref{CBOE}) the affine subspace $M(a)$ of the $s = -1/2$ real $L^2$-Sobolev space of $\mathbb{T}$: \begin{gather} \label{ClassicalPhaseSpace} M(a) = \left\{v(x)\colon ||v||_{-1/2}^2=2 \sum_{k=1}^{\infty} k^{-1} |V_k|^2 < \infty \ \text{and} \ V_0 = a\right\}. \end{gather}

\begin{defin} \label{ClassicalGFZDefinition} The Gardner--Faddeev--Zakharov Poisson bracket \begin{gather} \label{ClassicalBracket} \{ V_{-k}, V_{k'} \}_{\rm GFZ} = {\bf i} k \delta(k-k') \end{gather} is symplectic on the leaf $M(a)$ in (\ref{ClassicalPhaseSpace}) and defines a symplectic form $\upomega_{\rm GFZ}$ on $M(a)$. \end{defin}

 In geometric quantization of (\ref{CBOE}), we will use a compatible complex structure $J$ on $(M(a),\allowbreak \upomega_{\rm GFZ})$.

\begin{propo} \label{ClassicalPhaseSpaceCompatibility}The symplectic form $\upomega_{\rm GFZ}$ on $M(a)$ determined by \eqref{ClassicalBracket} is compatible with the complex structure on $M(a)$ defined by the spatial Hilbert transform $J$ with Fourier multiplier $J {\rm e}^{ \pm {\bf i} kx} = \mp {\bf i} {\rm e}^{\pm {\bf i}kx}$ which appears in \eqref{CBOE}. The compatibility of $\upomega_{\rm GFZ}$ and $J$ determines the real Sobolev inner product $g_{-1/2}$ on $M(a)$ of regularity $s=-1/2$ associated to the norm in \eqref{ClassicalPhaseSpace}. \end{propo}

 Proposition~\ref{ClassicalPhaseSpaceCompatibility} follows from the corresponding statements for $V_k \in \C$, as does the following:
 \begin{propo} \label{GFZformEXACT} The symplectic structure $\upomega_{\rm GFZ}$ on $M(a)$ is exact
\begin{gather*} \upomega_{\rm GFZ} = {\rm d} \upalpha_{\rm GFZ} \end{gather*} with canonical Liouville $1$-form $\upalpha_{\rm GFZ}$ given in the global coordinates of Fourier modes $V_k$ by \begin{gather} \label{1formFORMULA} \upalpha_{\rm GFZ} = 2 \sum_{k=1}^{\infty} k^{-1} \operatorname{Re}[V_k] d \operatorname{Im} [ V_k]. \end{gather}\end{propo}

\subsection{Classical Hamiltonian at criticality} \label{SECclassicalHamiltonianCRIT} The Benjamin--Ono equation (\ref{CBOE}) is Hamiltonian:
\begin{defin} \label{DEFclassicalBOhamiltonian} For $\ebar >0$, the classical periodic Benjamin--Ono Hamiltonian $O_3(\ebar)$ is \begin{gather} \label{ClassicalBOHamiltonian} O_3 ( \ebar)|_v= 3\sum_{h_1, h_2 = 0}^{\infty} V_{h_1} V_{h_2 - h_1} V_{-h_2} + 3\sum_{h=0}^{\infty} \big ( \ebar h - a \big ) V_{h} V_{-h} + a^3, \end{gather}
a partially-defined functional $v \mapsto O_3(\ebar)|_v$ on $(M(a), \upomega_{\rm GFZ})$ with $a = \int_0^{2\pi} v(x) \tfrac{{\rm d}x}{2 \pi}$ which generates the equations~(\ref{CBOE}) for $v(x,t; \ebar)$ as the Hamilton equations with respect to the bracket~(\ref{ClassicalBracket}).
\end{defin}

The a priori redundant constant term $a^3$ in (\ref{ClassicalBOHamiltonian}) emerges naturally in Proposition~\ref{Kerov337} below.

The classical Hamilton equations for (\ref{ClassicalBOHamiltonian}) in Definition~\ref{DEFclassicalBOhamiltonian} are formal: $M(a)$ is larger than the space $L^2(\mathbb{T})$ in which (\ref{CBOE}) is known to be well-posed \cite{GerardKappeler2019, Molinet}. From the perspective of dispersive equations, it is a coincidence that the symplectic space $M(a)$ of (\ref{CBOE}) corresponds to its critical regularity $s_c = -1/2$ as in Proposition~\ref{ClassicalPhaseSpaceCompatibility}. For discussion of criticality in~(\ref{CBOE}), see Saut~\cite{Saut2018}.

\subsection{Classical Lax operator as generalized Toeplitz operator} \label{SECclassicalLAX} We now review the definition of the classical Lax operator for Benjamin--Ono and its restriction to $L^2$ periodic Hardy space $H_{\bullet}$. Throughout the paper, the subscript ``$\bullet$'' denotes a construction defined by the Szeg\H{o} projection $\uppi_{\bullet}$.
\begin{defin} \label{HardyDef} Using the realization $\mathbb{T} = \{w \in \C \colon |w|=1\}$, the $L^2$-Hardy space $H_{\bullet}$ on~$\mathbb{T}$ is the Hilbert space closure of $\C[ w ]$ in $H=L^2(\mathbb{T})$. Equivalently, in terms of the {Szeg\H{o} projection} \begin{gather*} 
 ( \uppi_{\bullet} \Phi)(w_+) = \oint_{\mathbb{T}} \frac{ \Phi(w_-) }{ w_- - w_+} \frac{ {\rm d}w_- }{ 2 \pi {\bf i}} , \end{gather*} the periodic $L^2$-Hardy space $H_{\bullet}$ is the image of $\uppi_{\bullet}$ applied to $H = L^2(\mathbb{T})$.
\end{defin}

\begin{defin} \label{ClassicalBOLaxOperatorHERE} For $\ebar >0$ and bounded $v$, the classical Lax operator of~(\ref{CBOE}) is the unbounded self-adjoint operator $L_{\bullet}(v; \ebar)$ in Hardy space $H_{\bullet}$ defined to be the unique self-adjoint extension of the essentially self-adjoint operator \begin{gather}\label{ClassicalLaxMatrix} L_{\bullet}(v; \ebar) \big |_{\C[ w ]} =
 \begin{bmatrix}
(-0 \ebar + V_0 )& V_{-1} & V_{-2} & V_{-3} & \cdots & \\
V_{1} & (-1 \ebar+ V_0) & V_{-1} & V_{-2} & \ddots & \\
 V_{2} & V_{1} & (-2 \ebar+ V_0) & V_{-1} & \ddots & \\
 V_3 & V_2 & V_1 &(-3 \ebar + V_0 )& \ddots & \\
 \vdots & \ddots & \ddots & \ddots & \ddots \\
 \end{bmatrix} \end{gather}
 presented in the basis $|h \rangle = w^h$ for $h=0,1,2,\ldots$ of $\C[ w ]$. Equivalently, the Lax operator \begin{gather} \label{ClassicalLaxINTRINSIC} L_{\bullet}(v; \ebar)= - \ebar D_{\bullet} + L_{\bullet}(v) \end{gather} is the generalized Toeplitz operator of order $1$, where $L(v)$ is the operator of multiplication by~$v$, $L_{\bullet}(v) = \uppi_{\bullet} L(v) \uppi_{\bullet}$ is the Toeplitz operator of symbol $v$, and $D_{\bullet}$ acts by $D_{\bullet}| h \rangle = h | h \rangle$. \end{defin}

For background on Toeplitz operators, see Deift--Its--Krasovsky~\cite{DeiftItsKra}. The classical Lax ope\-rator $L_{\bullet}(v; \ebar)$ is essentially self-adjoint on $\C[ w ]$ since it is a bounded perturbation of $D_{\bullet} = \uppi_{\bullet}( w \partial_w ) \uppi_{\bullet}$ by a Toeplitz operator $L_{\bullet}(v)$ of bounded symbol~$v$. Note for $\ebar >0$ in~(\ref{ClassicalLaxINTRINSIC}), $- L_{\bullet}(v; \ebar)$ is elliptic.

\subsection{Classical Hamiltonian from classical Lax operator} \label{SECclassicalLAXHAMDICTIONARY}
By direct computation, one has: \begin{propo} \label{ClassicalBOHamiltonianVIALAX} The Hamiltonian~\eqref{ClassicalBOHamiltonian} on $(M(a), \upomega_{\rm GFZ})$ can be recovered as \begin{gather} \label{ClassicalRECOVERY} O_3(\ebar) = 3T_3^{\uparrow}(\ebar)- 3aT_2^{\uparrow}(\ebar) + a^3, \end{gather} where $a = \int_0^{2\pi} v(x) \tfrac{{\rm d}x}{2 \pi}$ and
\begin{gather} \label{ClassicalT2} T_2^{\uparrow}(\ebar) |_v = \big\langle 0 | L_{\bullet}(v; \ebar)^{2} | 0 \big\rangle, \\
\label{ClassicalT3} T_{3}^{\uparrow}(\ebar)|_v = \big \langle 0 | L_{\bullet}(v; \ebar)^{3} | 0 \big\rangle \end{gather}
are matrix elements of powers of the classical Lax operator~\eqref{ClassicalLaxMatrix} where $|0 \rangle = w^0 = 1\in H_{\bullet}$.
 \end{propo}

We use the same notation ``$\uparrow$'' as in multi-phase parameters $s_i^{\uparrow}$ in anticipation of results in Section~\ref{SECclassicalMP}.

\section[Classical Nazarov--Sklyanin hierarchy: dispersive action profiles]{Classical Nazarov--Sklyanin hierarchy:\\ dispersive action profiles} \label{SECclassicalNS}

In this section we recall the classical integrable hierarchy for~(\ref{CBOE}) from Nazarov--Sklyanin~\cite{NaSk2}, a collection of Poisson commuting Hamiltonians built from the Lax operator~(\ref{ClassicalLaxMatrix}). We also present the dispersive action profiles introduced by the author in~\cite{Moll1} which encode this classical hierarchy through spectral shift functions.

\subsection{Classical Nazarov--Sklyanin hierarchy} \label{SECclassicalNSpresentation} The following generalizes (\ref{ClassicalT2}) and (\ref{ClassicalT3}): \begin{defin} The classical Nazarov--Sklyanin hierarchy is the family of Hamiltonians
\begin{gather} \label{ClassicalNSHierarchy}T_{\ell}^{\uparrow}(\ebar)|_v = \big\langle 0 | L_{\bullet} (v; \ebar)^{\ell} | 0 \big\rangle \end{gather}
 defined as matrix elements of the $\ell$th power of the Lax operator \eqref{ClassicalLaxMatrix} for $|0 \rangle = w^0 =1 \in H_{\bullet}$. \end{defin}
\begin{theom}[Nazarov--Sklyanin \cite{NaSk2}] \label{ClassicalNStheorem} Restricting to the dense subspace in $M(a)$ of $v(x)$ with finite Fourier series in which all Hamiltonians \eqref{ClassicalNSHierarchy} are well-defined, for any $\ell_1, \ell_2 =0,1,2,3,\ldots $,
\begin{gather*}
\big\{T_{\ell_1}^{\uparrow}(\ebar) , T_{\ell_2}^{\uparrow}(\ebar) \big\}_{\rm GFZ} = 0 \end{gather*}
the \looseness=1 classical observables \eqref{ClassicalNSHierarchy} pairwise commute for the Gardner--Faddeev--Zakharov brac- \linebreak ket~\eqref{ClassicalBracket}.\end{theom}

Building upon their work \cite{NaSk1}, in \cite{NaSk2} Nazarov--Sklyanin define a classical Baker--Akhiezer function \begin{gather} \label{ClassicalNSbakerakhiezer} \Phi^{\rm BA}(u,w| v; \ebar) = \frac{1}{u - L_{\bullet}(v; \ebar)} | 0 \rangle \end{gather}of $u \in \C \setminus \R$, $w = {\rm e}^{{\bf i} x}$, and prove Theorem~\ref{ClassicalNStheorem} by showing Poisson-commutativity of all
\begin{gather} \label{ClassicalNSgenerating} T^{\uparrow}(u | \ebar)|_v = \left\langle 0 | \frac{1}{u - L_{\bullet}(v; \ebar)} | 0 \right\rangle, \end{gather}
which suffices since $T_{\ell}^{\uparrow}(\ebar)|_v$ in (\ref{ClassicalNSHierarchy}) is the coefficient of $u^{-\ell -1}$ in the expansion of (\ref{ClassicalNSgenerating}) at~$\infty$. In~$H_{\bullet}$, $\Phi^{\rm BA}(u,w | v; \ebar) = \sum\limits_{h=0}^{\infty} \Phi_h^{\rm BA}(u | v; \ebar) w^h$ and $T^{\uparrow}(u | \ebar)|_v = \Phi_0^{\rm BA}(u | v; \ebar)$, i.e., (\ref{ClassicalNSgenerating}) is the average value of~(\ref{ClassicalNSbakerakhiezer}) on~$\mathbb{T}$. Both (\ref{ClassicalNSbakerakhiezer}), (\ref{ClassicalNSgenerating}) are degenerations of quantum objects in Nazarov--Sklyanin~\cite{NaSk2} as we review in Section~\ref{SECquantumNSfinally}.

In recent work, G\'erard--Kappeler \cite{GerardKappeler2019} independently discovered the generating function (\ref{ClassicalNSgenerating}) and classical Baker--Akhiezer function (\ref{ClassicalNSbakerakhiezer}) from \cite{NaSk2} and gave a new proof of Theorem~\ref{ClassicalNStheorem} by constructing a new Lax pair for the classical Hamiltonian flow generated by (\ref{ClassicalNSgenerating}) with respect to the Poisson bracket (\ref{ClassicalBracket}).

\subsection{Embedded principal minor of classical Lax operator} \label{SECembeddedprincipalminor} The material in this section Section~\ref{SECembeddedprincipalminor} and the next section Section~\ref{SECkmk} is necessary in order to define in the following section Section~\ref{SECdispersiveactionprofiles} the dispersive action profiles for (\ref{ClassicalNSgenerating}) introduced by the author in~\cite{Moll1}. From now on, we assume $v$ bounded.

\begin{defin} \label{DEFLaxOpEmbeddedPrincipalMinor} The {embedded principal minor} $L_+(v; \ebar)$ of the Lax operator $L_{\bullet}(v; \ebar)$ is the unique self-adjoint extension to Hardy space $H_{\bullet}$ of the essentially self-adjoint operator in $\C[ w ]$ \begin{gather}\label{LaxEmbeddedPrincipalMinor} L_{+}(v; \ebar) \big |_{\C[ w ]} = \begin{bmatrix}
\ \ \ \ \ \ \ 0 \ \ \ \ \ \ \ & 0 & 0 & 0 & \cdots & \\
0& (-1 \ebar+ V_0) & V_{-1} & V_{-2} & \ddots & \\
0 & V_{1} & (-2 \ebar+ V_0) & V_{-1} & \ddots & \\
0& V_2 & V_1 &(-3 \ebar + V_0 )& \ddots & \\
 \vdots & \ddots & \ddots & \ddots & \ddots \\
 \end{bmatrix}. \end{gather} \end{defin}
 The operator (\ref{LaxEmbeddedPrincipalMinor}) is block diagonal $L_+(v; \ebar) = 0 \oplus L_+^{\perp}(v; \ebar)$ with respect to the orthogonal decomposition $H_{\bullet} = H_0 \oplus H_+$ of Hardy space where $H_0$ is the span of $| 0 \rangle = w^0=1$ and $L_+^{\perp}(v; \ebar)$ is the principal minor of $L_{\bullet}(v; \ebar)$. We review the role of (\ref{LaxEmbeddedPrincipalMinor}) in Nazarov--Sklyanin \cite{NaSk2} in Section~\ref{SECquantumNSfinally}.

\subsection{Essential self-adjointness and perturbation determinants} \label{SECkmk} The next results are from \cite[Section~3]{Moll1} and give a generalization of Cauchy's interlacing theorem from finite-rank to essentially self-adjoint operators:

\begin{propo}[\cite{Moll1}]\label{Paper1KMK}
Let $L_{\bullet}$ be a self-adjoint operator in a Hilbert space $H_{\bullet}$, $\psi_0 \in H_{\bullet}$ fixed, $H_0$
 the span of $\psi_0$, $H_+$ the orthogonal complement of $H_0$ in $H_{\bullet}$, and $L_+$ the embedded
 principal minor which is by definition block diagonal in $H_{\bullet} = H_0 \oplus H_+$ of the form
 $L_{+}= 0 \oplus L_+^{\perp}$ where~$L_+^{\perp}$ is the principal minor of $L_{\bullet}$ on $H_+$. If
 $L_{\bullet}$ is essentially self-adjoint on the $L_{\bullet}$ orbit of~$\psi_0$, then
 \begin{gather} \label{ClassicalOperatorKMK} \langle \psi_0 | \frac{1}{u - L_{\bullet}} | \psi_0 \rangle = T^{\uparrow}(u ) = \frac{1}{u} \cdot \frac{ \det_{H_{\bullet}}
 ( u - L_{+})}{ \det_{H_{\bullet}} (u - L_{\bullet} )}
 \end{gather}
for any $u \in \C \setminus \R$ the resolvent matrix element is
 a multiple of the perturbation determinant \begin{gather} \label{PertDetDef} \frac{ \det_{H_{\bullet}} (u - L_+) }
 { \det_{H_{\bullet}} (u -L_{\bullet}) } := \det\nolimits_{H_{\bullet}}
\big (\mathbbm{1} + (L_{\bullet}- L_+)( u - L_{\bullet})^{-1} \big ),
 \end{gather}
which is well-defined by the Fredholm determinant since $L_{\bullet} - L_+$ is rank $2$ hence trace class.
 \end{propo}

 \begin{coroll}[\cite{Moll1}] \label{Paper1KMKcorollary} Proposition {\rm \ref{Paper1KMK}} implies the relationship
 \begin{gather} \label{Paper1KMKformula} \int_{- \infty}^{+\infty} \frac{{\rm d} \tau^{\uparrow}_{\psi_0} (c | L_{\bullet})} { u - c} =T^{\uparrow}(u ) = \frac{1}{u} \cdot \exp \left(- \int_{- \infty}^{+\infty} \frac{\xi ( c | L_{\bullet}, L_+) {\rm d}c}{ u-c} \right ) \end{gather}
 between ${\rm d} \tau^{\uparrow}_{\psi_0} (c | L_{\bullet})$, the spectral measure of $L_{\bullet}$ at $\psi_0$, and $\xi(c |L_{\bullet}, L_+)$, the spectral shift function of $L_{\bullet}$ with respect to $L_+$ defined by equating~\eqref{Paper1KMKformula} with~\eqref{ClassicalOperatorKMK}.
 \end{coroll}

For general pairs of self-adjoint operators $L_{\bullet}$, $L_+$ whose difference $L_{\bullet} - L_+$ is trace class, the spectral shift function $\xi ( c | L_{\bullet}, L_+)$ is defined by identifying the right-hand side of~(\ref{Paper1KMKformula}) with the perturbation determinant in~(\ref{PertDetDef}). For background on spectral shift functions, see Birman--Pushnitski~\cite{BirPush}. For bounded $L_{\bullet}$, both Proposition~\ref{Paper1KMK} and Corollary~\ref{Paper1KMKcorollary} are due to Kerov~\cite{Ke1} in his theory of profiles and interlacing measures.

\subsection{Dispersive action profiles: definition} \label{SECdispersiveactionprofiles} Recall $L_{\bullet}(v; \ebar)$ from Section~\ref{SECclassicalLAX} and $L_+(v; \ebar)$ from Section~\ref{SECembeddedprincipalminor}.

 \begin{defin}[\cite{Moll1}] \label{DispersiveActionProfileDEF} Let $\mathbbm{1}_{[0, \infty)}(c)$ be the indicator function of $[0, \infty)$ and $\xi(c | L_{\bullet}(v; \ebar), L_+(v; \ebar))$ the spectral shift function of $L_{\bullet}(v; \ebar)$ with respect to $L_+(v; \ebar)$. The dispersive action profile $f(c | v; \ebar)$ is the unique function of $c \in \R$ so $f(c|v; \ebar) \sim | c- a|$ as $c \rightarrow \pm \infty$ for $a = \int_0^{2\pi} v(x) \tfrac{{\rm d}x}{2\pi}$ and \begin{gather} \label{Xi2f} \xi( c | L_{\bullet}(v; \ebar), L_+(v; \ebar)) = \tfrac{1}{2} (1 + f' (c |v; \ebar) ) - \mathbbm{1}_{[0, \infty)} (c) .\end{gather} \end{defin}

The dispersive action profile $f(c | v; \ebar)$ and the classical Nazarov--Sklyanin hierarchy (\ref{ClassicalNSHierarchy}) mutually determine each other: for bounded $v$, $L_{\bullet}(v; \ebar)$ is essentially self-adjoint on the orbit of~$|0 \rangle$ in~$H_{\bullet}$, hence by Proposition~\ref{Paper1KMK} and Corollary~\ref{Paper1KMKcorollary} one can match (\ref{ClassicalNSgenerating}) and~(\ref{Xi2f}) using~(\ref{Paper1KMKformula}). In particular, $f(c|v; \ebar)$ determines the classical periodic Benjamin--Ono Hamiltonian as follows:

\begin{propo} \label{Kerov337} For $v \in M(a)$, the dispersive action profile determines the energy \eqref{ClassicalBOHamiltonian} by \begin{gather*} 
O_3 (\ebar)|_v = \int_{- \infty}^{+\infty} c^3 \tfrac{1}{2} f''(c | v; \ebar) {\rm d}c. \end{gather*} \end{propo}

\begin{proof} Replacing the spectral shift function $\xi( c| v; \ebar)$ in (\ref{Paper1KMKformula}) with $f( c | v; \ebar)$ using (\ref{Xi2f}) gives \begin{gather}\label{IHaveThis} \sum_{\ell=0}^{\infty} T_{\ell}^{\uparrow}(\ebar)|_v u^{- \ell -1} = \exp \left( \int_{- \infty}^{+\infty} \log \left[ \frac{1}{u-c} \right] \tfrac{1}{2} f''(c | v; \ebar) {\rm d}c \right). \end{gather} By (\ref{ClassicalRECOVERY}), the coefficient of $u^{-4}$ in the logarithmic derivative of (\ref{IHaveThis}) in $u$ is the desired result. Indeed, the $a^3$ in (\ref{ClassicalBOHamiltonian}) is chosen to match formula (3.3.7) in~\cite{Ke1}.\end{proof}

\subsection{Dispersive action profiles: bands and gaps}\label{SECdispersiveactionprofilesBANDGAP}

We next recall from \cite{Moll1} why dispersive action profiles $f( c | v; \ebar)$ are piecewise-linear with slo\-pes~$\pm 1$. This motivates the next definition as in Section~\ref{SECIntro}.

\begin{defin}[\cite{Moll1}] \label{BandGapOriginalDef} The bands of dispersive action profiles $f(c| v; \ebar)$ are the closures of the connected intervals of $c \in \R$ in which $f(c | v ; \ebar)$ has slope $-1$. The gaps of dispersive action profiles $f( c| v; \ebar)$ are the interiors of the connected intervals of $c \in \R$ in which $f( c | v; \ebar)$ has slope~$+1$. \end{defin}

\begin{propo}[Boutet de Monvel--Guillemin \cite{DeMonvelGuillemin}] \label{GeneralizedToeplitzSpectrumDiscrete} $L_{\bullet}(v;\ebar)$ has discrete spectrum in $H_{\bullet}$ \begin{gather} \label{OriginalSpectrum} \cdots \leq C_2^{\uparrow}(v; \ebar) \leq C_1^{\uparrow}(v; \ebar) \leq C_0^{\uparrow}(v; \ebar) \end{gather} with eigenvalues $\big\{C_h^{\uparrow}(v; \ebar)\big\}_{i=0}^{\infty} $ bounded above with $-\infty$ as the only point of accumulation. \end{propo}

 \begin{coroll}[\cite{Moll1}] \label{InterlacingCorollary} For bounded $v$, the embedded principal minor $L_+(v; \ebar)$ defined by \eqref{LaxEmbeddedPrincipalMinor} has discrete spectrum with eigenvalues $\{0\} \cup \big\{C_h^{\downarrow}(v; \ebar)\big\}_{h=1}^{\infty}$ interlacing those of $L_{\bullet}(v; \ebar)$ \begin{gather} \label{Interlacing} \cdots \leq C_2^{\uparrow}(v; \ebar) \leq C_2^{\downarrow}(v; \ebar) \leq C_1^{\uparrow}(v; \ebar) \leq C_1^{\downarrow}(v; \ebar) \leq C_0^{\uparrow}(v; \ebar) \end{gather} and hence the dispersive action profile $f(c | v; \ebar)$ is piecewise-linear with slopes~$\pm 1$.
 \end{coroll}
\begin{proof} By Proposition~\ref{GeneralizedToeplitzSpectrumDiscrete}, (\ref{Xi2f}), and \cite[Section~2]{BirPush}, it is enough to show
\begin{gather*} 
| \xi (c | L_{\bullet}(v ; \ebar), L_+(v; \ebar) ) | \leq 1, \end{gather*} which holds since $L_{\bullet}(v; \ebar) - L_+(v; \ebar)$ has~1 positive and~1 negative eigenvalue.\end{proof}

\section[Classical integrability: G\'erard--Kappeler global action variables]{Classical integrability: G\'erard--Kappeler global action\\ variables} \label{SECclassicalINT}

 In this section we identify the global action variables of G\'erard--Kappeler \cite{GerardKappeler2019} in $L^2(\mathbb{T}) \cap M(a)$ with the gaps of dispersive action profiles $f(c | v; \ebar)$ from Definition~\ref{BandGapOriginalDef} introduced in \cite{Moll1} in the case of bounded $v$. We also state the characterization in G\'erard--Kappeler \cite{GerardKappeler2019} of these global action variables as integrals of the Liouville 1-form $\upalpha_{\rm GFZ}$ in (\ref{1formFORMULA}) along a basis of cycles $\Gamma_h^b(\ebar)$ of generically infinite-dimensional tori $\Lambda^b(\ebar)$ parametrized by profiles $b$. This relationship between profiles and classical actions plays a key role in our proof of Theorem~\ref{MAINTHEOREM} in Section~\ref{SECquantumMP}.

\subsection{Principal minor and shift relation} Inspection of (\ref{LaxEmbeddedPrincipalMinor}) immediately gives the shift relation:

\begin{lema} \label{ShiftRelation}Under the shift operator identifying $H_{+} \cong H_{\bullet}$ the subspace $H_+$ of periodic $L^2$ Hardy space $H_{\bullet}$ spanned by $\big\{w^h\big\}_{h=1}^{\infty}$ with Hardy space itself, the action of the principal minor
\begin{gather*} L_+^{\perp} (v; \ebar) \big |_{w \C[ w ] } \cong \begin{bmatrix}
 - \ebar & 0 & 0 & \ddots & \\
0& - \ebar & 0 & \ddots & \\
0 & 0&-\ebar & \ddots & \\
 \vdots & \ddots & \ddots & \ddots \\ \end{bmatrix} + \begin{bmatrix}
 (0 \ebar+ V_0) & V_{-1} & V_{-2} & \ddots & \\
V_{1} & (-1 \ebar+ V_0) & V_{-1} & \ddots & \\
 V_2 & V_1 &(-2 \ebar + V_0 )& \ddots & \\
 \vdots & \ddots & \ddots & \ddots \\ \end{bmatrix} \end{gather*}
 in the dense subspace of $H_+$ is unitarily equivalent to that of a shifted classical Lax operator
 \begin{gather*}
 {\rm Id}_{H_{\bullet}} + L_{\bullet}(v; \ebar) .\end{gather*}
 As a consequence, the eigenvalues $C_h^{\downarrow} (v; \ebar)$ of the embedded principal minor $L_{+}(v; \ebar)$ of the classical Lax operator $L_{\bullet}(v; \ebar)$ can be calculated from those of $L_{\bullet}(v; \ebar)$ by the shift relation \begin{gather} \label{ShiftRelationSpectrum} C_h^{\downarrow}(v; \ebar) = - \ebar + C_{h-1}^{\uparrow}(v; \ebar). \end{gather}
\end{lema}

\subsection{Interlacing property and simplicity of spectrum}
 As a first application of the shift relation in Lemma~\ref{ShiftRelation}, we give a new short proof of Proposition~2.1 in G\'erard--Kappeler~\cite{GerardKappeler2019} for bounded $v$.
\begin{propo}[\cite{GerardKappeler2019}] \label{Simplicity} $C_{h}^{\uparrow} (v; \ebar) \leq - \ebar + C_{h-1}^{\uparrow}(v;\ebar)$ hence $L_{\bullet}(v; \ebar)$ has simple spectrum in $H_{\bullet}$. \end{propo}

\begin{proof} Use formula (\ref{ShiftRelationSpectrum}) to write the interlacing property (\ref{Interlacing}) of Corollary~\ref{InterlacingCorollary} as \begin{gather*} \cdots \leq C_h^{\uparrow} (v; \ebar) \leq - \ebar + C_{h-1}^{\uparrow}(v; \ebar) \leq C_{h-1}^{\uparrow}(v; \ebar) \leq \cdots, \end{gather*} which implies the bound and that $C_h^{\uparrow} (v; \ebar)< C_{h-1}^{\uparrow} (v; \ebar)$ all inequalities in (\ref{OriginalSpectrum}) are strict.\end{proof}

\subsection{Bands and spatial periodicity conditions II} Next, we derive for generic $v$ a counterpart to Proposition~\ref{MultiPhasePeriodicityConditions} for multi-phase $v$. \begin{propo} \label{SPECTRALPeriodicityConditions} If $v(x)$ is $2\pi$-periodic in $x$, then the dispersive action profile $f(c| v; \ebar)$ has band lengths that are all positive integer multiples of $\ebar>0$.\end{propo}
\begin{proof} By Definitions~\ref{DispersiveActionProfileDEF} and~\ref{BandGapOriginalDef}, the bands of the dispersive action profile are unions of consecutive intervals $\big[C_h^{\downarrow}(v; \ebar), C_{h-1}^{\uparrow}(v; \ebar)\big]$ each of length $\ebar>0$ by the shift relation~(\ref{ShiftRelationSpectrum}).\end{proof}

\subsection{Gaps as G\'erard--Kappeler global action variables} \label{SECsubsecBANDGAPdictionary} We now give a description of gaps:

\begin{propo} \label{GapInnit} Gaps of dispersive action profiles $f( c | v; \ebar)$ have the form $\big(C_h^{\uparrow}(v; \ebar), C_h^{\downarrow}(v; \ebar)\big)$. \end{propo}

\begin{proof} Follows from the interlacing inequalities in Corollary~\ref{InterlacingCorollary} and Proposition~\ref{SPECTRALPeriodicityConditions}.\end{proof}

After our study of gaps of dispersive action profiles in \cite{Moll1}, the same gaps were shown to be global action variables in a comprehensive analysis by G\'erard--Kappeler \cite{GerardKappeler2019} who found global action-angle variables for~(\ref{CBOE}) posed in the space of real $L^2$ functions on $\mathbb{T}$. We now present G\'erard--Kappeler's description of the gaps in the following theorem, which is a strict subset of Theorem~1 from~\cite{GerardKappeler2019} and stated here using relations between constructions in \cite{GerardKappeler2019, Moll1, NaSk2} established in Sections~\ref{SECclassicalNSpresentation} and~\ref{SECdispersiveactionprofiles}.

\begin{theom}[G\'erard--Kappeler \cite{GerardKappeler2019}] \label{GKexcerpt} For any fixed $v \in M(a) \cap L^2(\mathbb{T})$ and $\ebar>0$,
\begin{itemize}\itemsep=0pt
\item {\rm [Tori]} The phase space $M(a) \cap L^2(\mathbb{T})$ is foliated by Liouville tori \begin{gather} \label{GKTori} \Lambda^{b}(\ebar) = \{ v \colon f(c | v; \ebar) = b(c)\} \end{gather}
consisting of all $v$ whose dispersive action profiles are equal to a fixed profile~$b(c)$.
\item {\rm [Cycles]} The map $\varphi$ which takes $v$ to its classical Baker--Akhiezer function~\eqref{ClassicalNSbakerakhiezer}
\begin{gather*} 
\varphi\colon \ v \longrightarrow \Phi^{\rm BA}(u,w| v; \ebar) \end{gather*} is injective and has an inverse $\varphi^{-1}$ defined on the image of $\varphi$. A smooth global basis of cycles $\big\{{\Gamma}_h^b(\ebar)\big\}_{h=1}^{\infty}$ on the Liouville tori $\Lambda^b(\ebar)$ is given by the pushforward along $\varphi^{-1}$ of the cycle in the space of meromorphic functions in $u$ which rotates the residue at $u=C_h^{\uparrow}(v; \ebar)$.
\item {\rm [Actions]} For $\upalpha_{\rm GFZ}$ the Liouville $1$-form from Proposition~{\rm \ref{GFZformEXACT}}, the classical actions \begin{gather} \label{GKActions} \oint_{{\Gamma}_h^b(\ebar)} \upalpha_{\rm GFZ} =2\pi \ebar \big| C_h^{\uparrow}(v; \ebar) - C_h^{\downarrow}(v; \ebar) \big| \end{gather} around the cycles ${\Gamma}_h^b(\ebar)$ are $2\pi \ebar>0$ multiples of the length of the gap $\big(C_h^{\uparrow}(v; \ebar), C_h^{\downarrow}(v; \ebar)\big)$.
\end{itemize}
\end{theom}

\section[Quantum periodic Benjamin--Ono: Hamiltonian and Lax operator]{Quantum periodic Benjamin--Ono: Hamiltonian \\ and Lax operator} \label{SECquantumBO}

In this section we quantize the classical Benjamin--Ono equation (\ref{CBOE}) for $v$ periodic in $x$ by choosing $J$-holomorphic quantizations of the classical phase space, Hamiltonian, and Lax operator in Section~\ref{SECclassicalBO}.
\subsection[Quantum state space as Fock--Sobolev space from Segal--Bargmann construction]{Quantum state space as Fock--Sobolev space\\ from Segal--Bargmann construction} \label{SUBSECsegalbargmann}

Recall from Proposition~\ref{ClassicalPhaseSpaceCompatibility} that the spatial Hilbert transform $J$ defines a complex structure on the classical phase space $M(a)$ compatible with the metric $g_{-1/2}$ associated to the $L^2$-Sobolev norm of regularity $s= -1/2$. As a state space for the quantization of (\ref{CBOE}) we choose the Fock space of $J$-holomorphic functionals on $M(a)$ given by the Segal--Bargmann construction. To emphasize its dependence on the regularity $s=-1/2$, we may refer to this Fock space as Fock--Sobolev space.
\begin{defin} \label{QuantumPhaseSpaceDEF} For $a \in \R$ and $\hbar >0$, the Fock--Sobolev space is the complex Hilbert space
\begin{gather*} 
\overline{F}(a) = L^2_{\text{J-hol}} (M(a), \rho_{-1/2, \hbar}) \end{gather*} of $J$-holomorphic functionals on $M(a)$ square-integrable against the Segal--Bargmann Gaussian weight $\rho_{-1/2, \hbar}$ given by the standard Gaussian on $(M(a), \upomega_{\rm GFZ}, J, g_{-1/2})$ of variance~$\hbar>0$. \end{defin}

The Segal--Bargmann construction is standard in quantization. For background, see \cite[Remark~4.4]{Janson} and \cite[Section~9]{Woodhouse}. As an alternative to Definition~\ref{QuantumPhaseSpaceDEF}, we also recall that one can define the Fock--Sobolev space indirectly as a Hilbert space completion:
\begin{propo} The Fock--Sobolev space $\overline{F}(a)$ is the completion of the polynomial ring \begin{gather*} 
F(a) = \C[V_1, V_2, \ldots ] \end{gather*} in the infinitely-many Fourier modes $V_k$ from~\eqref{FourierDEF} with inner product $\langle \cdot, \cdot \rangle_{\hbar}$ defined by requiring $V_{\mu}:=V_1^{d_1} V_2^{d_2} \cdots$ with $d_k \in \{0,1,2,\ldots\}$ to be orthogonal with norm $||V_{\mu } ||_{\hbar}^2 = \prod\limits_{k=1}^{\infty} (\hbar k)^{d_k} d_k! $. \end{propo}

The Fourier modes $V_{\pm k}$ are functionals on $M(a)$ which satisfy $\overline{V}_{\pm k} = {V}_{\mp k}$ and are also canonical coordinates on $(M(a),\upomega_{\rm GFZ})$ as seen in (\ref{ClassicalBracket}). The quantum analogs of $V_{\pm k}$ are also well-known:

\begin{defin} \label{DEFquantumLADDERops} The creation and annihilation operators are the mutually-adjoint operators of multiplication $\widehat{V}_k = V_k$ and differentiation $\widehat{V}_{-k} = \hbar k \frac{\partial}{\partial V_k}$, respectively, in $\big(\overline{F}(a), \langle \cdot, \cdot \rangle_{\hbar}\big)$ satisfying
\begin{gather*} 
\big[ \widehat{V}_{-k} , \widehat{V}_k \big] = \hbar k \delta(k-k') \end{gather*} the quantum canonical commutation relations of the same form as the classical relations~\eqref{ClassicalBracket}. \end{defin}

\subsection{Quantum Hamiltonian at criticality} We now quantize the classical Benjamin--Ono equation (\ref{CBOE}) in $M(a)$ by replacing $V_{\pm k}$ in the classical Hamiltonian (\ref{ClassicalBOHamiltonian}) by $\widehat{V}_{\pm k}$ from Definition~\ref{DEFquantumLADDERops}.

\begin{defin} \label{DEFquantumBOhamiltonian} For $\hbar$ independent of $\ebar$ and $a$, $\widehat{V}_{\pm k}$ in Definition \ref{DEFquantumLADDERops}, and $\widehat{V}_0 = a$, the \textit{quantum periodic Benjamin--Ono equation} is the quantum Hamiltonian system in Fock--Sobolev space $\big(\overline{F}(a), \langle \cdot, \cdot \rangle_{\hbar}\big)$ determined by the quantum Hamiltonian defined without normal ordering by \begin{gather} \label{QuantumBOHamiltonian} \widehat{O}_3 ( \ebar, \hbar) = 3 \sum_{h_1, h_2 = 0}^{\infty} \widehat{V}_{h_1} \widehat{V}_{h_2 - h_1} \widehat{V}_{-h_2} + 3\sum_{h=0}^{\infty} \big ( \ebar h - a \big ) \widehat{V}_{h} \widehat{V}_{-h} + a^3 .\end{gather}
\end{defin}

The procedure of directly replacing classical canonically conjugate modes by their quantum analogs usually results in an ill-defined operator in Hilbert space that must be regularized by normal ordering. An important feature of the formula~(\ref{ClassicalBOHamiltonian}) is that substituting $V_k \mapsto \widehat{V}_k$ in~(\ref{ClassicalBOHamiltonian}) results in~(\ref{QuantumBOHamiltonian}) which is well-defined {without normal ordering}.

\subsection{Quantum Lax operator as generalized Fock-block Toeplitz operator} \label{SECquantumLAX} Let $\mathcal{F}$ be a vector space over $\mathbb{C}$ and $H_{\bullet}$ the Hardy space of Definition~\ref{HardyDef}. Recall that \textit{block Toeplitz operators} on $\mathcal{F} \otimes H_{\bullet}$ are Toeplitz operators in $H_{\bullet}$ whose matrix elements are linear operators on $\mathcal{F}$. For background on block Toeplitz operators, see Section~10 in Deift--Its--Krasovsky~\cite{DeiftItsKra}. Using material from Section~\ref{SECclassicalLAX}, we define the quantum Lax operator for~(\ref{CBOE}) as a generalized block Toeplitz operator in $\overline{F}(a) \otimes H_{\bullet}$. Since $\mathcal{F} = \overline{F}(a)$ is Fock space, we call it a generalized ``Fock-block'' Toeplitz operator.
\begin{defin} \label{QuantumBOLaxOperatorHERE} For $\ebar >0$ and $\hbar>0$, the quantum Benjamin--Ono Lax operator is the self-adjoint operator $L_{\bullet}(\widehat{v} ( \cdot, \hbar) ; \ebar)$ in $\overline{F}(a) \otimes H_{\bullet}$ realized as the unique extension of
 \begin{gather}\label{QuantumLaxMatrix} L_{\bullet}(\widehat{v} ( \cdot, \hbar); \ebar) \big |_{F(a) \otimes \C[ w ]} = \begin{bmatrix}
\big({-}0 \ebar + \widehat{V}_0 \big)\! & \widehat{V}_{-1} & \widehat{V}_{-2} & \widehat{V}_{-3} & \cdots \\
\widehat{V}_{1} & \big({-}1 \ebar+ \widehat{V}_0\big)\! & \widehat{V}_{-1} & \widehat{V}_{-2} & \ddots \\
 \widehat{V}_{2} & \widehat{V}_{1} & \big({-}2 \ebar+ \widehat{V}_0\big)\! & \widehat{V}_{-1} & \ddots \\
 \widehat{V}_3 & \widehat{V}_2 & \widehat{V}_1 &\big({-}3 \ebar + \widehat{V}_0 \big)\! & \ddots \\
 \vdots & \ddots & \ddots & \ddots & \ddots \end{bmatrix}\!,\!\!\! \end{gather} where $\widehat{V}_{\pm k}$ are from Definition~\ref{DEFquantumLADDERops} and $\widehat{V}_0$ acts by the scalar $a$. Equivalently,
 \begin{gather*} 
 L_{\bullet}(\widehat{v} ( \cdot, \hbar); \ebar)= \mathbbm{1}_{\overline{F}(a)} \otimes ( - \ebar D_{\bullet} ) + L_{\bullet}(\widehat{v}(\cdot, \hbar) )\end{gather*}
 is the generalized Fock-block Toeplitz operator of order $1$ whose symbol is the affine $\widehat{\mathfrak{gl}_1}$ current
\begin{gather} \label{CURRENTdef} \widehat{v}(x; \hbar) = \sum_{k \in \mathbb{Z}} \widehat{V}_k {\rm e}^{- {\bf i} k x} \end{gather} defined by replacing $V_{\pm k} \mapsto \widehat{V}_{\pm k}$ in the Fourier series (\ref{FourierDEF}) for the classical field $v$. \end{defin}

 The Fock-block matrix (\ref{QuantumLaxMatrix}) is essentially self-adjoint in $F(a) \otimes \C[ w ]$ due to the Szeg\H{o} projections. Indeed, since $F(a) = \C[V_1, V_2, \ldots]$, the matrix (\ref{QuantumLaxMatrix}) preserves the dense subspace in question \begin{gather*} L_{\bullet}(\widehat{v} ( \cdot, \hbar); \ebar) \colon \ F(a) \otimes \C [ w ] \rightarrow F(a) \otimes \C[ w ] \end{gather*} and commutes with the operator $ \sum\limits_{k=1}^{\infty} \widehat{V}_k \widehat{V}_{-k} \otimes \mathbbm{1} + \mathbbm{1} \otimes D_{\bullet}$ with finite-dimensional eigenspaces so essential self-adjointness follows by Nussbaum's criteria. As a corollary, for any $\Phi^{{\rm out}}, \Phi^{{\rm in}} \in \C[ w ]$, \begin{gather*} 
 \big\langle \Phi^{{\rm out}} | L_{\bullet}(\widehat{v} ( \cdot, \hbar); \ebar)^{\ell} | \Phi^{{\rm in}} \big\rangle \colon \ F(a) \rightarrow F(a) \end{gather*} the matrix element of $\ell$th powers of the quantum Lax matrix preserves $F(a)$. By contrast, without $\uppi_{\bullet}$, $\ell$th powers of the operator $L(\widehat{v}( \cdot, \hbar))$ which multiplies by the current~(\ref{CURRENTdef}) with zero mode $a$ and level $\hbar$ are ill-defined on $F(a)$, a well-known issue in the theory of vertex algebras that is discussed in Kac~\cite{Kac0}.

\subsection{Quantum Hamiltonian from quantum Lax operator} \label{SECquantumLAXHAMDICTIONARY} As in Section~\ref{SECclassicalLAXHAMDICTIONARY}, we have: \begin{propo} \label{QuantumBOHamiltonianVIALAX} The quantum periodic Benjamin--Ono Hamiltonian can be recovered as \begin{gather*} \widehat{O}_3(\ebar, \hbar) = 3\widehat{T}_3^{\uparrow}(\ebar, \hbar)- 3a \widehat{T}_2^{\uparrow}(\ebar, \hbar) + a^3 \end{gather*} in the Fock--Sobolev space $(\overline{F}(a), \langle \cdot, \cdot \rangle_{\hbar})$ where \begin{gather} \label{QuantumT2} \widehat{T}_2^{\uparrow}(\ebar, \hbar) = \big\langle 0 | L_{\bullet}(\widehat{v}( \cdot, \hbar); \ebar)^{2} | 0 \big\rangle, \\ \label{QuantumT3} \widehat{T}_{3}^{\uparrow}(\ebar, \hbar) = \big\langle 0 | L_{\bullet}(\widehat{v}(\cdot, \hbar); \ebar)^{3} | 0 \big\rangle \end{gather} are matrix element of powers of the quantum Lax operator \eqref{QuantumLaxMatrix} for $|0 \rangle = w^0 = 1 \in H_{\bullet}$.
 \end{propo}

\section[Quantum stationary states: Jack functionsand Abanov--Wiegmann renormalization]{Quantum stationary states: Jack functions \\ and Abanov--Wiegmann renormalization} \label{SECquantumINT}

In this section we show how the renormalization $\ebar \rightarrow \varepsilon_1$ in~(\ref{Renormalization}) of the classical dispersion coefficient in Abanov--Wiegmann~\cite{AbWi1} is implicit in the known realization of Jack functions as quantum periodic Benjamin--Ono stationary states.

\subsection{Quantum periodic Benjamin--Ono stationary states} \label{SECjackstrue} Recall the definition of partitions. \begin{defin} \label{RegularPartitionDEF} A partition $\lambda$ is a weakly-decreasing sequence $0 \leq \cdots \leq \lambda_2 \leq \lambda_1$ of non-negative integers $\lambda_h \in \mathbb{N} = \{0,1,2,\ldots\}$ labeled by $h=1,2,3,\ldots$ so that $\sum\limits_{h=1}^{\infty} \lambda_h < \infty.$ \end{defin}

Partitions $\lambda$ index the pure quantum stationary states of our quantization of (\ref{CBOE}):
\begin{propo} The quantum periodic Benjamin--Ono Hamiltonian $\widehat{O}_3(\ebar, \hbar)$ in Fock--Sobolev space $\overline{F}(a)$ is self-adjoint with discrete spectrum indexed by partitions $\lambda$ and eigenfunctions \begin{gather} \label{QPBOStationaryStates} P_{\lambda,a} (V_1, V_2, \ldots | \ebar, \hbar) = \sum_{ \mu} \chi_{\lambda}^{\mu}(\ebar, \hbar) V_{\mu}, \end{gather} which are polynomials in $V_k$ independent of $a$, i.e., finite linear combinations of $V_{\mu} = V_1^{d_1} V_2^{d_2} \cdots$ indexed by partitions $\mu$ with $d_i = \# \{ j \colon \mu_j = i\}$ so that $\chi_{\lambda}^{\mu}(\ebar, \hbar)= 0$ if $\sum_i \lambda_i \neq \sum_j \mu_j$. \end{propo}

\begin{proof} By the definition of the quantum Lax operator (\ref{QuantumLaxMatrix}), (\ref{QuantumT2}) is independent of $\ebar$ and is
\begin{gather*} \widehat{T}_2 (\hbar ) = \sum_{k=1}^{\infty} \widehat{V}_k \widehat{V}_{-k} .\end{gather*} $\widehat{T}_2(\hbar)$ acts diagonally on $V_{\mu}$ with eigenvalue $\hbar \sum_j \mu_j$. By direct calculation, it commutes \begin{gather*} \big[ \widehat{O}_3(\ebar, \hbar), \widehat{T}_2(\hbar)\big] = 0 \end{gather*} with $\widehat{O}_3(\ebar, \hbar)$ in (\ref{QuantumBOHamiltonian}), hence $\widehat{O}_3(\ebar, \hbar)$ preserves the finite-dimensional eigenspaces of $\widehat{T}_2(\hbar)$ spanned by $V_{\mu}$ with fixed $\sum_j \mu_j$. Since~(\ref{QuantumBOHamiltonian}) is symmetric under the exchange $\widehat{V}_k \leftrightarrow \widehat{V}_{-k}$ which are mutual adjoints in Fock space, $\widehat{O}_3(\ebar, \hbar)$ is self-adjoint on the finite-dimensional eigenspaces of $\widehat{T}_2(\hbar)$. The result then follows from the spectral theorem.
\end{proof}

\subsection{Jack functions and Abanov--Wiegmann renormalization} \label{SECconventions}
Several authors discovered that (\ref{QPBOStationaryStates}) are \textit{Jack functions}. We state this result in light of \cite{AbWi1}.
\begin{theom}[Stanley \cite{Stanley}, Polychronakos \cite{Poly1995}, Awata--Matsuo--Odake--Shiraishi \cite{AwMtOdSh}] Using the conventions for power sum symmetric functions $p_k$ and the Jack parameter $\alpha$ from Macdo\-nald~{\rm \cite{Mac}}, the quantum periodic Benjamin--Ono stationary states~\eqref{QPBOStationaryStates} are Jack functions with
 \begin{gather} \label{Convention1} V_k = (- \varepsilon_2) p_k, \\ \label{Convention2} \alpha = {\varepsilon_1}/{(- \varepsilon_2)}\end{gather}
 in which $\varepsilon_1 = \varepsilon_1(\ebar, \hbar)$ in~\eqref{Renormalization} and $\varepsilon_2 = \varepsilon_2( \ebar, \hbar)$ in~\eqref{RenormalizationCompanion} are defined from the coefficients of dispersion $\ebar$ and quantization $\hbar$ by the renormalization~\eqref{Renormalization} found by Abanov--Wiegmann~{\rm \cite{AbWi1}}.
\end{theom}

The reduction to Schur functions at $\alpha = 1$ is $\ebar = 0$ when (\ref{CBOE}) has no dispersion term is in accordance with quantizations by Dubrovin \cite{Dubrovin2014} and Karabali--Polychronakos~\cite{KarabaliPolychronakos2014}.

\section[Quantum Nazarov--Sklyanin hierarchy: anisotropic partition profiles]{Quantum Nazarov--Sklyanin hierarchy:\\ anisotropic partition profiles}\label{SECquantumNS}

In this section we present the solution of the quantization problem for the classical integrable system (\ref{CBOE}) posed in $M(a)$ by Nazarov--Sklyanin \cite{NaSk2}. We also present the exact formula in Nazarov--Sklyanin \cite{NaSk2} for the spectrum of their quantum integrable hierarchy in terms of anisotropic partition profiles of anisotropy $(\varepsilon_2, \varepsilon_1)$ from Definition~\ref{AnisotropicProfileDEF}.

\subsection{The quantization problem for classical integrable systems} \label{SECquantizationproblem}

 Given a smooth symplectic manifold $(M, \upomega)$, we recall four types of quantizations of Poisson subalgebras $\mathsf{A} \subset C^{\infty}(M, \R)$ and state the quantization problem for classical integrable systems. The definitions below are all standard. We refer the reader to the survey of Faddeev~\cite{FaddeevDEF} and to Dubrovin \cite[Definition 1.1]{Dubrovin2014}.
\begin{defin} \label{DEFQUANTdef} For formal $\hbar$, a deformation quantization $\star_{\hbar}$ of $\mathsf{A}$ is a phase space star product \begin{gather*} O_1 \star_{\hbar} O_2 = \sum_{p=0}^{\infty} \hbar^p B_{p} (O_1, O_2) \end{gather*} defined for $O_1, O_2 \in \mathsf{A}$ by bilinear operators $B_p\colon \mathsf{A} \times \mathsf{A} \rightarrow \mathsf{A} \otimes \C$ of order $(p,p)$ for which
\begin{itemize}\itemsep=0pt
\item[\rm (i)] if $\mathsf{A}$ has a unit $1$, $1 \star_{\hbar} O = O \star_{\hbar} 1 = O$ is a two-sided identity,
\item[\rm (ii)] to leading-order, $B_0 (O_1, O_2) = O_1 O_2$, so $\star_{\hbar}$ deforms the commutative product $\star_0$,
\item[\rm (iii)] the anti-symmetric part $B_1(O_1, O_2)^- = \{O_1, O_2\}$ is the Poisson bracket in $(M, \upomega)$ for \begin{gather} \label{SymAntiSym} B_{1}(O_1, O_2)^{\pm} = \tfrac{1}{2} \big ( B_1(O_1, O_2) \pm B_1(O_2, O_1) \big ) .\end{gather} \end{itemize} \end{defin}

\begin{defin} For $\hbar>0$, an operator quantization $Q_{\hbar}$ of $\mathsf{A}$ is a unitary representation of $\star_{\hbar}$ from Definition~\ref{DEFQUANTdef}, i.e., a choice of a Hilbert space of quantum states $(\mathcal{H}, \langle \cdot, \cdot \rangle)$ and a map
\begin{gather*} Q_{\hbar}\colon \ \mathsf{A} \rightarrow {\bf i} \mathfrak{u} ( \mathcal{H}, \langle \cdot, \cdot \rangle) \end{gather*} to the space ${\bf i} \mathfrak{u} ( \mathcal{H}, \langle \cdot, \cdot \rangle)$ of self-adjoint operators in $(\mathcal{H}, \langle \cdot, \cdot \rangle)$ so the pullback $\star_{\hbar}^Q$ of multiplication \begin{gather*} Q_{\hbar}\big( O_1 \star_{\hbar}^Q O_2\big) := Q_{\hbar}(O_1) \cdot Q_{\hbar}(O_2) \end{gather*} of self-adjoint operators is a deformation quantization. Below we write $Q_{\hbar}(O) = \widehat{O}^Q(\hbar)$. \end{defin}

\begin{defin} \label{GeometricQuantizationDEF} Given an almost complex structure $J$ on $M$ compatible with $\upomega$ and associated Riemannian metric $g$, a $J$-holomorphic quantization of $\mathsf{A}$ is an operator quantization $Q_{\hbar}$ so that \begin{gather*} B_1^{Q} (O_1, O_2)^+ = g( \nabla_g O_1, \nabla_g O_2) \end{gather*} the symmetric part of the first bidifferential defined by~(\ref{SymAntiSym}) is the inverse metric $g^{-1}$. \end{defin}

\begin{defin} \label{CommutativeQuantizationDEF} Given a Poisson-commutative subalgebra $\mathsf{T} \subset \mathsf{A}$, i.e., for all $T_{\ell_1}, T_{\ell_2} \in \mathsf{T}$
\begin{gather*} \{ T_{\ell_1}, T_{\ell_2} \} = 0,\end{gather*} a $\mathsf{T}$-commutative quantization of~$\mathsf{A}$ is an operator quantization $Q$ so that for all $T_{\ell_1}, T_{\ell_2} \in \mathsf{T}$,\begin{gather*} \big[\widehat{T}_{\ell_1}^Q, \widehat{T}^Q_{\ell_2} \big] = 0 \end{gather*}
 the quantization of $\mathsf{T}$ is a commutative subalgebra of self-adjoint operators in ${\bf i} \mathfrak{u} ( \mathcal{H}, \langle \cdot, \cdot \rangle_{\hbar})$. \end{defin}

The quantization problem for integrable systems is to construct an $\mathsf{T}$-commutative quantization of $\mathsf{A} \subset C^{\infty}(M, \R)$ where $\mathsf{T}$ is the Poisson commutative subalgebra spanned by a classical integrable hierarchy in $(M, \upomega)$. For further discussion, see \cite[Definition~1.1]{Dubrovin2014}.

\subsection{Quantum Nazarov--Sklyanin hierarchy} \label{SECquantumNSfinally}

From the perspective of Section~\ref{SECquantizationproblem}, the main result in Nazarov--Sklyanin \cite{NaSk1} is an explicit solution to the quantization problem for the classical integrable hierarchy (\ref{ClassicalNSHierarchy}). Recall that the quantum Hamiltonian $\widehat{O}_3(\ebar, \hbar)$ and Lax operator $L_{\bullet}( \widehat{v}( \cdot, \hbar) ; \ebar)$ in Section~\ref{SECquantumBO} are defined without normal ordering by replacing $V_{\pm k} \rightarrow \widehat{V}_{\pm k} (\hbar)$ in formulas from Section~\ref{SECclassicalBO}.

\begin{defin} \label{QuantumNShierarchy} For $\ell=0,1,2,3,\ldots$ the quantum Nazarov--Sklyanin hierarchy \begin{gather} \label{QuantumNSHierarchy} \widehat{T}_{\ell}^{\uparrow} (\ebar, \hbar) = \big\langle 0 | L_{\bullet}(\widehat{v}( \cdot, \hbar) , \ebar)^{\ell} | 0 \big\rangle \end{gather} are the self-adjoint operators in $\overline{F}(a)$ defined without normal ordering by replacing \linebreak $V_{\pm k} \rightarrow \widehat{V}_{\pm k} (\hbar)$ in formula (\ref{ClassicalNSHierarchy}) for the classical Nazarov--Sklyanin hierarchy. At $\ell=2,3$, (\ref{QuantumNSHierarchy}) is~(\ref{QuantumT2}),~(\ref{QuantumT3}).
\end{defin}

In Section~\ref{SECquantumLAX} we saw that matrix elements of $\ell$th powers of the quantum Lax operator such as~(\ref{QuantumNSHierarchy}) are well-defined in $\overline{F}(a)$. The particular matrix elements~(\ref{QuantumNSHierarchy}) have a remarkable property:

\begin{theom}[Nazarov--Sklyanin \cite{NaSk2}] \label{QuantumNStheorem} For any $\ell_1, \ell_2 = 0,1,2,3,\ldots$, \eqref{QuantumNSHierarchy} commute \begin{gather} \label{TheTrueThing} \big[ \widehat{T}^{\uparrow}_{\ell_1}(\ebar, \hbar) , \widehat{T}^{\uparrow}_{\ell_2}(\ebar, \hbar) \big] = 0 .\end{gather}
\end{theom}

The map $Q_{\hbar}^{\rm NS} \colon T_{\ell}^{\uparrow}(\ebar) \mapsto \widehat{T}^{\uparrow}(\ebar, \hbar)$ defined on the Poisson-commutative subalgebra $\mathsf{T}(\ebar)$ generated by $\big\{T_{\ell}^{\uparrow}(\ebar)\big\}_{\ell=0}^{\infty}$ defines both a $\mathsf{T}(\ebar)$-commutative and $J$-holomorphic quantization of~$\mathsf{T}(\ebar)$, where $J$ is the spatial Hilbert transform. Theorem~\ref{QuantumNStheorem} implies Theorem~\ref{ClassicalNStheorem} by taking the $\hbar$-expansion. The hierarchy (\ref{QuantumNSHierarchy}) and relation (\ref{TheTrueThing}) was also found by Sergeev--Veselov \cite{SergVes, SergVes2}.

 \textit{Conventions}: We verify that our presentation of the quantum Nazarov--Sklyanin hierarchy~(\ref{QuantumNSHierarchy}) is equivalent to the original one in \cite{NaSk2}. First, use (\ref{Convention1}), (\ref{Convention2}) to change conventions in~\cite{NaSk2} (which are those of Macdonald \cite{Mac}) to ours (which we discussed in Section~\ref{SECnekrasov}). We have $\alpha -1 = \ebar / (- \varepsilon_2)$ and $p_k^* = \alpha k \frac{\partial}{\partial p_k}$ in formula (5.2) of \cite{NaSk2} is our $\widehat{V}_{-k} = \widehat{V}^{\dagger}_k = \hbar k \frac{\partial}{\partial V_k}$ from Section~\ref{SUBSECsegalbargmann}. Next, scaling the operator in formula (6.1) of \cite{NaSk2} by $1/(- \varepsilon_2)$ yields the principal minor of
 \begin{gather}\label{QuantumLaxMatrixEmbeddedPrincipalMinor} L_{+} (\widehat{v} ( \cdot, \hbar); \ebar) \big |_{F(a) \otimes \C[w]} = \begin{bmatrix}
0& 0 & 0& 0 & \cdots & \\
0 & (-1 \ebar+ \widehat{V}_0) & \widehat{V}_{-1} & \widehat{V}_{-2} & \ddots \\
0 & \widehat{V}_{1} & (-2 \ebar+ \widehat{V}_0) & \widehat{V}_{-1} & \ddots \\
0 & \widehat{V}_2 & \widehat{V}_1 &(-3 \ebar + \widehat{V}_0 )& \ddots \\
 \vdots & \ddots & \ddots & \ddots & \ddots \end{bmatrix} \end{gather} the embedded principal minor $L_+(\widehat{v}( \cdot, \hbar); \ebar)$ of our quantum Lax operator $L_{\bullet} (\widehat{v}( \cdot, \hbar); \ebar)$ in (\ref{QuantumLaxMatrix}). Note in \cite{NaSk2} one assumes $\widehat{V}_0=a = 0.$ In this way, Theorem 2 in \cite{NaSk2} asserts commutativity of
 \begin{gather} \label{QuantumNSHierarchyALT} \widehat{T}^{\downarrow}_{\ell} ( \ebar, \hbar) := \big\langle 0 | L_{\bullet} (\widehat{v}( \cdot, \hbar); \ebar) L_+(\widehat{v}( \cdot, \hbar); \ebar)^{\ell} L_{\bullet} (\widehat{v}( \cdot, \hbar); \ebar) | 0 \big\rangle
 \end{gather} for $\ell = 0,1,2,3,\ldots$, as can be seen by using our (\ref{QuantumLaxMatrix}), (\ref{QuantumLaxMatrixEmbeddedPrincipalMinor}) to rewrite formula~(6.4) in \cite{NaSk2}. Finally, commutativity of (\ref{QuantumNSHierarchyALT}) is equivalent to commutativity of (\ref{QuantumNSHierarchy}) since for $a=0$, $u \in \C \setminus \R$, \begin{gather} \label{Clutch}
 \big(u- \widehat{T}^{\downarrow}(u | \ebar, \hbar)\big)^{-1} = \widehat{T}^{\uparrow}(u| \ebar, \hbar) \end{gather}
 the series $\widehat{T}^{\uparrow}(u | \ebar, \hbar) = \sum\limits_{\ell=0}^{\infty} u^{-\ell -1} \widehat{T}_{\ell}^{\uparrow}(\ebar, \hbar)$ is the resolvent of $\widehat{T}^{\downarrow}(u | \ebar, \hbar) = \sum\limits_{\ell=0}^{\infty} u^{-\ell -1} \widehat{T}_{\ell}^{\downarrow}(\ebar, \hbar)$. To prove (\ref{Clutch}), use (\ref{QuantumBOLaxOperatorHERE}) to expand $\widehat{T}^{\uparrow}_{\ell}(\ebar, \hbar)$ in~(\ref{QuantumNSHierarchy}) as a sum over paths of length $\ell$ in $\{0,1,2,\ldots\}$ which start and end at $0$, look for the first step where the path returns to~$0$, and collect terms to form~(\ref{QuantumNSHierarchyALT}) from which the equivalence of presentations follows.

\begin{Remark} The proof of Theorem \ref{QuantumNStheorem} in Nazarov--Sklyanin \cite{NaSk2} relies on their earlier work~\cite{NaSk1} where they construct a different family of commuting operators $\widehat{A}(u| \ebar, \hbar)$ diagonalized on Jack functions~(\ref{QPBOStationaryStates}). The $\widehat{A}(u| \ebar, \hbar)$ in \cite{NaSk1} serve to define a \textit{quantum Baker--Akhiezer function} in formula~(7.1) of~\cite{NaSk2}. While we do not make use of the quantum Baker--Akhiezer function below, for completeness let us mention that in the notation of \cite{NaSk2} the classical limit is $\alpha \rightarrow 0$, hence by formulas~(5.5), (7.1) in~\cite{NaSk2} as $\alpha \rightarrow 0$ the quantum Baker--Akhiezer function degenerates to the classical Baker--Akhiezer function~(\ref{ClassicalNSbakerakhiezer}). Recall also that (\ref{ClassicalNSbakerakhiezer}) was independently discovered by G\'erard--Kappeler~\cite{GerardKappeler2019} and plays a role in their results which we reviewed in Theorem~\ref{GKexcerpt}.
\end{Remark}

\subsection{Partitions and anisotropic partition profiles}

\begin{lema} \label{AnisotropicLemma} For any $r_2 < 0 < r_1$ and $a \in \R$ fixed, there is a bijection between partitions \begin{gather}\label{Bijection} \lambda \longleftrightarrow f_{\lambda}( c-a | r_2, r_1 ) \end{gather} and anisotropic partition profiles of anisotropy $(r_2, r_1)$ centered at $a \in \R$ from Definition~{\rm \ref{AnisotropicProfileDEF}}. \end{lema}
\begin{proof} The rectangles which tile the region~(\ref{TheRegion}) below an anisotropic partition profile are grouped in rows of positive slope indexed by $h=1,2,3,\ldots$ starting from the right. The count~$\lambda_h$ of the number of rectangles $R(r_2, r_1)$ in the $h$th row defines the necessary bijection~(\ref{Bijection}).
\end{proof}

 To illustrate (\ref{Bijection}), the partition $\lambda$ in Fig.~\ref{anisotropicPOPFIGURE} is $\cdots \leq 1 \leq 1 \leq 1 \leq 4 \leq 4 \leq 4 \leq 4$.

\subsection{Spectrum of the quantum Nazarov--Sklyanin hierarchy}
Without relying on knowledge of the classical multi-phase solutions of (\ref{CBOE}) nor on semi-classical approximation, not only did Nazarov--Sklyanin \cite{NaSk2} solve the quantization problem for (\ref{CBOE}) by constructing the hierarchy (\ref{QuantumNSHierarchy}), they also found the exact quantum spectrum (Hamiltonian eigenvalues) of this hierarchy at the Jack functions $P_{\lambda, a}(V | \ebar, \hbar)$ (\ref{QPBOStationaryStates}) in terms of the Jack parameter $\alpha$ and the parts $\lambda_h$ of the partition $\lambda$. One striking feature of the spectrum in \cite{NaSk2} is that it can be presented in terms of the anisotropic partition profile $f_{\lambda}(c-a | \varepsilon_2, \varepsilon_1)$ from Lemma~\ref{AnisotropicLemma} using the conventions from Section~\ref{SECconventions}:
\begin{theom}[Nazarov--Sklyanin~\cite{NaSk2}] \label{QuantumNSSPECTRUMtheorem} For any partition $\lambda$, $a \in \R$, and $\ebar, \hbar>0$, the eigenvalue
\begin{gather*} \widehat{T}^{\uparrow}(u; \ebar, \hbar) P_{\lambda;a}( V | \ebar, \hbar) = T^{\uparrow}_{\lambda;a}(u | \ebar, \hbar) P_{\lambda;a}( V | \ebar, \hbar)\end{gather*} of the generating function of the quantum periodic Benjamin--Ono hierarchy \eqref{QuantumNSHierarchy} defined by
\begin{gather*} 
\widehat{T}^{\uparrow}(u; \ebar, \hbar) = \sum_{\ell=0}^{\infty} \widehat{T}^{\uparrow}_{\ell}(\ebar, \hbar) u^{- \ell-1} \end{gather*}
at a Jack function $P_{\lambda;a}(V | \ebar, \hbar) \in \overline{F}(a)$ in the Fock--Sobolev space associated to the classical phase space $(M(a), \upomega_{\rm GFZ})$ with zero mode $a = \int_0^{2\pi} v(x) \tfrac{{\rm d}x}{2 \pi}$ is given for $u \in \C \setminus \R$ by \begin{gather} \label{NSMegaFormula} T^{\uparrow}_{\lambda;a}(u | \varepsilon_1, \hbar) = \exp \left( \int_{- \infty}^{+\infty} \log \left[ \frac{1}{u-c} \right] \tfrac{1}{2} f_{\lambda}''(c-a | \varepsilon_2, \varepsilon_1) {\rm d}c \right ), \end{gather} where $f_{\lambda}( c-a | \varepsilon_2, \varepsilon_1)$ is the anisotropic partition profile of anisotropy $(\varepsilon_2, \varepsilon_1)$ centered at $a \in \R$. \end{theom}

Using Section~\ref{SECquantumNSfinally}, one can match (\ref{NSMegaFormula}) with the eigenvalue formula in \cite[Section~6]{NaSk2} by writing each using $\prod_{\square \in \lambda} S(c(\square) | \varepsilon_2, \varepsilon_1)$ where $c(\square)$ is the content of the box $\square$ in $\lambda$ and $S(c | \varepsilon_2, \varepsilon_1) = \frac{(c + \varepsilon_2)(c + \varepsilon_1)}{ c ( c+ \varepsilon_2 + \varepsilon_1)}.$

\section{Multi-phase solutions: finite-gap conditions} \label{SECclassicalMP}

In this section we recall our result from \cite{Moll1} that multi-phase solutions (\ref{BOmPhaseSolution}) of (\ref{CBOE}) are finite-gap. We also discuss why this result agrees with a subsequent classification of finite-gap solutions by G\'erard--Kappeler \cite{GerardKappeler2019} in order to apply their computations of classical action integrals for arbitrary $v$ -- which we presented in Theorem~\ref{GKexcerpt} -- to multi-phase $v = v^{\vec{s}, \vec{\chi}}(x, t; \ebar)$ in Section~\ref{SECquantumMP}.

\subsection{Multi-phase solutions are finite-gap}
 The multi-phase solutions (\ref{BOmPhaseSolution}) have not appeared at all in Sections~\ref{SECclassicalBO},~\ref{SECclassicalNS},~\ref{SECclassicalINT},~\ref{SECquantumBO},~\ref{SECquantumINT} and~\ref{SECquantumNS} above. The following result from \cite{Moll1} establishes the relevance of the constructions and results from these previous sections to the study of multi-phase solutions.
\begin{theom}[{\cite[Theorem 1.2.1]{Moll1}}] \label{FiniteGapResult} For any $\vec{s} \in \R^{2n+1}$ from~\eqref{BOMultiPhaseParameters} and $\vec{\chi} = (\chi_n, \ldots, \chi_1) \in \R^n$, \begin{gather} \label{FiniteGapResultFormula} f\big(c | v^{\vec{s}, \vec{\chi} } ( \cdot ; \ebar); \ebar\big) = f( c| \vec{s})\end{gather} the dispersive action profiles $f( c | v ; \ebar)$ from Definition~{\rm \ref{DispersiveActionProfileDEF}} of the multi-phase solutions $v = v^{\vec{s}, \vec{\chi}}(x, t; \ebar)$ \eqref{BOmPhaseSolution} are equal to the Dobrokhotov--Krichever profiles $f(c | \vec{s})$ from Definition~{\rm \ref{MultiPhaseProfileDEF}}. In particular, for multi-phase $v = v^{\vec{s}, \vec{\chi}}$, only finitely-many gaps $\big(s_i^{\uparrow}, s_i^{\downarrow}\big)$ are non-empty. \end{theom}

In the following proposition, we refine the spectral description of $f(c|\vec{s})$ in Theorem~\ref{FiniteGapResult}.

\begin{propo} \label{IKnowTheBandsAndGaps} For $N_i \in \mathbb{Z}_+$ in \eqref{BOMultiPhasePeriodicConditions} and $h_i = N_i + \cdots + N_1$, non-zero gaps in \eqref{FiniteGapResultFormula} are \begin{gather} \label{IKnowTheGapsFormula} \big(s_i^{\uparrow}, s_i^{\downarrow}\big) = \big( C_{h_i}^{\uparrow}(v^{\vec{s}, \vec{\chi}}; \ebar) , C_{h_i}^{\downarrow}\big(v^{\vec{s}, \vec{\chi}}; \ebar\big) \big) \end{gather} and the bands $\big[s_i^{\downarrow}, s_{i-1}^{\uparrow}\big]$ are unions of $N_i$ intervals $\big[C_{h}^{\downarrow}(v; \ebar), C_{h-1}^{\uparrow}(v; \ebar)\big]$ each of size $\ebar >0$. \end{propo}

\begin{proof} For any $v$, by Proposition \ref{GapInnit}, the gaps in the dispersive action profile $f( c | v; \ebar)$ must have the form $\big(C_h^{\uparrow}(v; \ebar), C_h^{\downarrow}(v; \ebar)\big)$ for some $h$. For multi-phase $v = v^{\vec{s}, \vec{\chi}}$, by Proposition~\ref{MultiPhasePeriodicityConditions}, the bands $\big[s_i^{\downarrow}, s_{i-1}^{\uparrow}\big]$ in the Dobrokhotov--Krichever profile $f(c | \vec{s})$ have length $\ebar N_i$ for $N_i \in \mathbb{Z}_+$. The identification of gaps (\ref{IKnowTheGapsFormula}) follows from our identification of profiles in Theorem~\ref{FiniteGapResult}.\end{proof}

\subsection{Multi-phase solutions from G\'erard--Kappeler classification} As a consequence of the identification of gaps in \cite{GerardKappeler2019, Moll1} in Section~\ref{SECsubsecBANDGAPdictionary}, Theorem~\ref{FiniteGapResult} can also be seen to follow from a recent classification of finite gap solutions:

\begin{theom}[{\cite[Theorem 3]{GerardKappeler2019}}] \label{GKFINITE} $v(x)$ has a dispersive action profile $f( c| v; \ebar)$ with finitely-many gaps of non-zero length if and only if it is of the form
\begin{gather} \label{TheForm} v(x) = C_0 - 2 \ebar \operatorname{Im} \partial_x \log \tau_v\big({\rm e}^{{\bf i} x}\big) \end{gather} for $C_0 \in \R$ and $\tau_v$ a polynomial in $w = {\rm e}^{{\bf i} x}$ whose zeroes all lie outside the closed unit disk. \end{theom}

\begin{proof}[Proof that Theorem~\ref{GKFINITE} implies Theorem~\ref{FiniteGapResult}] \sloppy By the Dobrokhotov--Krichever for\-mu\-la~(\ref{BOmPhaseSolution}), the classical multi-phase solutions of Satsuma--Ishimori~\cite{SatsumaIshimori1979} are of the form~(\ref{TheForm}) for $\tau_v ({\rm e}^{{\bf i} x}) = \det M_n^{\vec{s}, \vec{\chi}}( x,t; \ebar)$ where the entries of the $n \times n$ matrix $M_n^{\vec{s}, \vec{\chi}}$ in (\ref{BOmPhaseSolutionMatrix}) are polynomials in $w = {\rm e}^{{\bf i} x}$. By Lemma~1.1 in Dobrokhotov--Krichever~\cite{DobrokhotovKrichever}, the $n$ eigenvalues of $M_n^{\vec{s}, \vec{\chi}}$ in $w = {\rm e}^{{\bf i} x}$ lie outside the closed unit disk, so by Theorem~\ref{GKFINITE}, (\ref{BOmPhaseSolution}) are finite-gap.
\end{proof}

\begin{Remark} $w = {\rm e}^{{\bf i} x}$ in Theorem~\ref{GKFINITE} matches $\C[w]$ in Definition~\ref{HardyDef} since in~\cite{GerardKappeler2019} there is a~relationship between $\Phi^{\rm BA}(u,w | v; \ebar)$ in~(\ref{ClassicalNSbakerakhiezer}) and a $\tau_v$ as in~(\ref{TheForm}) for all $v \in M(a) \cap L^2(\mathbb{T})$. In~\cite{Moll1}, we verified such a relationship for multi-phase $v$: at $v=v^{\vec{s}, \vec{\chi}}$, $\Phi^{\rm BA}(u, w | v; \ebar)$ comes from the Baker--Akhiezer function on the singular spectral curves in Dobrokhotov--Krichever~\cite{DobrokhotovKrichever} defined from two solutions to two non-stationary Schr\"{o}dinger equations whose time-dependent potentials determine $\tau_v$ in~(\ref{TheForm}).
\end{Remark}

\section{Multi-phase solutions: Bohr--Sommerfeld conditions} \label{SECquantumMP}

In this section we prove Theorem~\ref{MAINTHEOREM} in 7 Steps. In Steps 1--5 we derive formula (\ref{ClassicalActionsRESULT}). In Step~6 we derive formula (\ref{BOMultiPhaseBohrSommerfeldConditions}), completing the proof of Part I. In Steps 7--9 we prove Part~II.
\subsection{Step 1: 1-phase case of (\ref{ClassicalActionsRESULT})}
For the $1$-phase Benjamin--Ono periodic traveling wave (\ref{BO1PhaseSolution}) with $s_1^{\uparrow} < s_1^{\downarrow} < s_0^{\uparrow}$, the $n=1$, $i=1$ case of~(\ref{ClassicalActionsRESULT}) can be directly computed from the closed formula~(\ref{BO1PhaseSolution}) and the series formula~(\ref{1formFORMULA}) for the Liouville 1-form to give
\begin{gather} \label{ClassicalActionsRESULTfor1phase} \oint_{\upgamma_{1,1}^{\vec{s}} (\ebar)} \upalpha_{\rm GFZ}
 = 2\pi \ebar \big|s_1 ^{\uparrow} - s_1^{\downarrow} \big|.\end{gather}
We omit the calculation of (\ref{ClassicalActionsRESULTfor1phase}) since it is also the $n=1$ case of formula~(\ref{SecondExpansion}) in Step 3 below. Note that Step 3 below is logically independent of Step 2 below, so we can use (\ref{ClassicalActionsRESULTfor1phase}) in Step 2.
\subsection{Step 2: Asymptotic validity of (\ref{ClassicalActionsRESULT})}

As in (\ref{BOMultiPhasePeriodicConditions}), for $\vec{N} = (N_n, \ldots, N_1 )\in \mathbb{Z}_+^n$ define \begin{gather} B^{{\rm reg}}_n = \big\{ \vec{s} \in \R^{2n+1} \colon s_n^{\uparrow} < s_n^{\downarrow} < \cdots s_1^{\uparrow} < s_1^{\downarrow} < s_0^{\uparrow} \big\}, \nonumber\\
\label{MegaBase} B^{{\rm reg}}_{\vec{N}; n} (a; \ebar) = \left\{ \vec{s} \in B^{{\rm reg}}_n \colon a = \sum_{i=0}^n s_i^{\uparrow} - \sum_{i=1}^n s_i^{\downarrow} \ \text{and for all $i$}, \ \big|s_i^{\downarrow} - s_{i-1}^{\uparrow} \big| = \ebar N_i \right\} .\end{gather} Each $B^{{\rm reg}}_{\vec{N}; n} (a; \ebar)$ is diffeomorphic to $\R_{\geq 0}^n = [0, \infty)^n$ with coordinates the $n$ gap lengths $|s_{i}^{\uparrow} - s_{i}^{\downarrow}|$. As will be important below, note that both $\R_{\geq}^n = [0,\infty)^n$ and $\R_{>}^n = (0,\infty)^n$ are simply-connected. If all gap lengths $\big|s_j^{\uparrow} - s_j^{\downarrow}\big| \rightarrow \infty$ diverge, i.e., in \textit{any} limit to $\infty$ in $B^{{\rm reg}}_{\vec{N}; n} (a; \ebar) \cong \R^n_{>0}$, we claim \begin{gather} \label{KeyAsymptoticRelation} \oint_{\upgamma_{i,n}^{\vec{s}} (\ebar)} \upalpha_{\rm GFZ} \sim 2\pi \ebar \big| s_i^{\uparrow} - s_i^{\downarrow}\big|\end{gather} that (\ref{ClassicalActionsRESULT}) holds asymptotically. The proof of (\ref{KeyAsymptoticRelation}) is as follows: as all gap lengths diverge, the off-diagonal entries of the matrix (\ref{BOmPhaseSolutionMatrix}) vanish, hence the logarithmic derivative of the determinant in (\ref{BOmPhaseSolution}) splits into a sum indexed by $j=1, \ldots, n$. Since the cycle $\upgamma_{i,n}^{\vec{s}} (\ebar)$ varies only~$\chi_i$, and since $\chi_i$ appears only in the term with $j=i$, the action integral is asymptotically given by the $n=1$ case in Step~1. The asymptotic relation~(\ref{KeyAsymptoticRelation}) appears in the proof of Lemma~1.1 in Dobrokhotov--Krichever \cite{DobrokhotovKrichever} and is the regime in which the multi-phase solution becomes a~linear superposition of $1$-phase solutions.

\subsection{Step 3: Cycle decomposition of (\ref{ClassicalActionsRESULT}) and G\'erard--Kappeler actions} For $i=1, \ldots, n$, consider the cycle $\upgamma_{i,n}^{\vec{s}}(\ebar)$ in (\ref{ClassicalActionsRESULT}) defined from the formula (\ref{BOmPhaseSolution}) of Dobrokho\-tov--Krichever \cite{DobrokhotovKrichever}. Since the multi-phase profile $f( c| \vec{s})$ of Definition~\ref{MultiPhaseProfileDEF} is independent of $\chi_i$ in~(\ref{BOmPhaseSolution}), the cycle $\upgamma_{i,n}^{\vec{s}}(\ebar)$ lies in a torus (\ref{GKTori}) from Theorem~\ref{GKexcerpt} of G\'erard--Kappeler \cite{GerardKappeler2019}: \begin{gather} \label{Lies} \upgamma_{i,n}^{\vec{s}}(\ebar) \subset \Lambda^{f( \cdot | \vec{s})}(\ebar), \end{gather}
 which is $\Lambda^b(\ebar)$ for $b (c)= f( c | \vec{s})$. Decompose $\upgamma_{i,n}^{\vec{s}}(\ebar)$ relative to the basis of cycles $\Gamma_h^{f( \cdot | \vec{s})}(\ebar)$ in the torus $\Lambda^{f( \cdot | \vec{s})}(\ebar)$ from Theorem \ref{GKexcerpt}. By Theorem~\ref{FiniteGapResult}, $f( c| \vec{s})$ is the dispersive action profile $f( c | v^{\vec{s}, \vec{\chi}} ( \cdot ; \ebar) ; \ebar)$ of the multi-phase solution. By Proposition~\ref{IKnowTheBandsAndGaps}, the torus $\Lambda^{f( \cdot | \vec{s})}$ in (\ref{Lies}) is $n$ dimensional and has as a basis the cycles $\Gamma_{h_j}^{f( \cdot | \vec{s})} (\ebar)$ from Theorem~\ref{GKexcerpt} indexed by $h_j = N_j + \cdots +N_1$ for $j=1, \ldots, n$. By (\ref{Lies}), for classes $[ \ \cdot \ ]$ in first homology $H_1 ( \Lambda^{f ( \cdot | \vec{s})} , \mathbb{Z})$ of a fixed fiber, we have \begin{gather} \label{FirstExpansion} \big[\upgamma_{i,n}^{\vec{s}}(\ebar)\big] = \sum_{j=1}^n C^{\vec{s}}_{j,i}(\ebar) \big[ {\Gamma}_{h_j}^{f( \cdot | \vec{s})}(\ebar)\big], \end{gather}
 where $C_{j,i}^{\vec{s}}(\ebar) \in \mathbb{Z}$ depend a priori on $\vec{s} \in B^{{\rm reg}}_{\vec{N}; n} (a; \ebar)$. Pairing (\ref{FirstExpansion}) with the $1$-form $\upalpha_{\rm GFZ}$, using the actions~(\ref{GKActions}) of G\'erard--Kappeler~\cite{GerardKappeler2019} from Theorem~\ref{GKexcerpt}, and Proposition~\ref{IKnowTheBandsAndGaps}, we get
\begin{gather} \label{SecondExpansion} \oint_{\upgamma_{i,n}^{\vec{s}} (\ebar)} \upalpha_{\rm GFZ}= 2\pi\ebar \sum_{j=1}^n C^{\vec{s}}_{j,i}(\ebar) \big|s_j^{\uparrow} - s_j^{\downarrow} \big|. \end{gather}

\subsection{Step 4: Cycle decomposition of (\ref{ClassicalActionsRESULT}) is constant} We claim that the coefficients in (\ref{FirstExpansion}) \begin{gather} \label{Constancy} C_{j,i}^{\vec{s}}(\ebar) = C_{j,i}(\ebar) \end{gather}
 do not depend on $\vec{s}$. This is a short but crucial step in the proof. (\ref{Constancy}) follows since (i) the fibration in Theorem \ref{GKexcerpt} is smooth and (ii) $B^{{\rm reg}}_{\vec{N}; n} (a; \ebar)$ in (\ref{MegaBase}) is simply-connected (since $\R_{>0}^n$ is), so the completely integrable system associated to the $n$-phase solutions is monodromy-free.

\subsection{Step 5: Evaluation of (\ref{ClassicalActionsRESULT})} By (\ref{SecondExpansion}) and (\ref{Constancy}), to prove (\ref{ClassicalActionsRESULT}) it suffices to prove
\begin{gather} \label{ThirdExpansion} \sum_{j=1}^n C_{j,i}(\ebar) \big|s_j^{\uparrow} - s_j^{\downarrow} \big| = \big| s_i^{\uparrow} - s_i^{\downarrow} \big|, \end{gather} which is equivalent to the $n$ relations $C_{j,i}(\ebar) = \delta (i-j)$. Restating the asymptotic relation (\ref{KeyAsymptoticRelation}) from Step 2 using the decomposition (\ref{ThirdExpansion}), in any limit in which all $|s_j^{\uparrow} - s_j^{\downarrow}| \rightarrow \infty$, we know \begin{gather*} \sum_{j=1}^{n} C_{j,i}(\ebar ) \big| s_j^{\uparrow} - s_j^{\downarrow} \big| \sim \big|s_i^{\uparrow} - s_i^{\downarrow} \big|. \end{gather*}
 Taking $n$ different limits in which all gaps diverge but the $j$th gap grows faster than the others gives the desired $n$ relations $C_{j,i} (\ebar) = \delta(i-j)$. Indeed, (\ref{ThirdExpansion}) is linear in $\big|s_i^{\uparrow} - s_i^{\downarrow}\big|$ with constant coefficients so the coefficients are determined by~(\ref{KeyAsymptoticRelation}). $C_{j,i}^{\vec{s}}(\ebar) = \delta(i-j)$ in~(\ref{SecondExpansion}) gives~(\ref{ClassicalActionsRESULT}).

\subsection{Step 6: Regular Bohr--Sommerfeld conditions} We now prove Part I of Theorem~\ref{MAINTHEOREM}. For completeness, we first recall the definition of the regular Bohr--Sommerfeld conditions on the actions of a Liouville integrable system and comment on the geometric assumptions taken in our definition.
\begin{defin} \label{RegularBSDEF} For a classical Liouville integrable system in $(M, \upomega)$ of $\dim_{\mathbb{R}} M = 2n$ with
\begin{itemize}\itemsep=0pt
\item $\upomega = {\rm d} \upalpha$ an exact symplectic form with Liouville 1-form $\upalpha$,
\item ${T}\colon M \rightarrow B$ the associated moment map to a simply-connected base $B$ of $\dim_{\mathbb{R}} B = n$,
\item $\Lambda^{{b}} = {T}^{-1}( \vec{b})$ Lagrangian fibers given by the Liouville tori,
\item $\upgamma_i^{{b}}$ a basis of cycles of the tori $\Lambda^{{b}}$ indexed by $i=1,\ldots, n$,
\item ${b} \in B^{{\rm reg}} \subset B$ a regular value of $T$,
\end{itemize}
 the regular Bohr--Sommerfeld conditions on ${b}$ are the $n$ conditions for $i=1,\ldots, n$ given by \begin{gather} \label{RegularBSConditionsGeneral} \oint_{\upgamma_i^{{b}}} \upalpha = 2 \pi \hbar N_i', \end{gather} where $N_i' \in \mathbb{Z}_+$ is a positive integer and $\hbar>0$ is a dimensionless real parameter of quantization.\end{defin}

 Bohr--Sommerfeld conditions -- and associated semi-classical approximations of quantum spectra~-- have been long studied in mathematical physics. For background, see Takhtajan \cite[Section~6.3]{TakhtajanBOOK}, V\~{u}~Ngoc \cite[Section~5]{SVNbohrsommerfeld2001}, and Woodhouse \cite[Section~8.4]{Woodhouse}. Definition~\ref{RegularBSDEF} is a special case of the definition of Bohr--Sommerfeld leaves of general real polarizations of $M$ (whose Lagrangian leaves $\Lambda$ are not necessarily tori).

In practice, the assumption $[\upomega] = 0$ that the symplectic form is exact is often weakened to $[\upomega] \in H^2(M; \mathbb{Z})$, thus trading $\upomega = d \upalpha$ for the realization $\upomega = \textbf{F}_{\nabla}$ of the symplectic form as the curvature 2-form of a connection $\nabla$ on a line bundle $\mathbb{L} \rightarrow M$. In this setting, one reformulates the Bohr--Sommerfeld conditions as the requirement that the holonomy group of the flat connec\-tion~$\nabla |_{\Lambda}$ is trivial. For simply-connected $B$, a result of Guillemin--Sternberg~\cite{GuilleminSternbergBSRESULT} guarantees that this more general definition specializes to our Definition~\ref{RegularBSDEF} above.

Next, we argue that the assumptions in Definition \ref{RegularBSDEF} apply to our problem. At first glance, this seems impossible: the multi-phase profiles $b(c) = f(c | \vec{s})$ are certainly not regular values of the moment map which takes $v$ to its dispersive action profile $b(c) = f(c | v ; \ebar)$ (or, equivalently, the gap lengths). As we saw in Step 3, the tori $\Lambda^{f( \cdot | \vec{s})}$ explored by multi-phase solutions has real-dimension $n$, but generic tori $\Lambda^b(\ebar)$ in (\ref{GKTori}) are infinite-dimensional (generic $v$ are infinite-gap). However, in Part I we are to neglect the infinitely-many transverse directions in phase space to~$\Lambda^{f( \cdot | \vec{s})}$, an assumption that will allow us to use the regular Bohr--Sommerfeld conditions. For $\vec{N} = (N_n, \ldots, N_1) \in \mathbb{Z}_+^n$, consider the $n$ spectral indices $h_n > \cdots > h_1$ for $h_j = N_j + \cdots + N_1$ from Proposition \ref{IKnowTheBandsAndGaps} and define
\begin{gather} \label{FiniteTotalSpace} M_{\vec{N}; n} (a; \ebar) = \big\{ v \in M(a) \cap L^2(\mathbb{T}) \colon h \not \in \{h_n, \ldots, h_1 \} \Rightarrow \big|C_{h}^{\uparrow} (v; \ebar) - C_h^{\downarrow} (v; \ebar) \big|= 0 \big\} .\end{gather}
$M_{\vec{N}; n}(a; \ebar)$ is the space of $v$ whose gap lengths $\big|C_h^{\uparrow}(v; \ebar) - C_h^{\downarrow}(v; \ebar) \big| \geq 0$ for $h \in \{h_n, \ldots, h_1\}$ are either positive or zero. By Theorems~\ref{GKexcerpt} and~\ref{FiniteGapResult}, $M_{\vec{N}; n} (a; \ebar) $ in (\ref{FiniteTotalSpace}) is the phase space of an integrable subsystem of (\ref{CBOE}) associated to multi-phase solutions with moment map
\begin{gather} \label{FiniteMomentMap} M_{\vec{N}; n}(a; \ebar) \rightarrow B_{\vec{N}; n}(a; \ebar) \end{gather}
 given by taking the multi-phase profile (or, equivalently, the gap lengths). Notice that the base $B_{\vec{N}; n}(a; \ebar) \cong \R_{\geq 0}^{n}$ is simply-connected, and that $B^{{\rm reg}}_{\vec{N}; n} (a; \ebar) \subset B_{\vec{N}; n}(a; \ebar)$ from~(\ref{MegaBase}) is the open set of regular values of (\ref{FiniteMomentMap}) associated to all positive gap lengths which is also simply-connected. In particular, the case $n=0$ is indeed counted in Part I of Theorem~\ref{MAINTHEOREM}, as it corresponds to
 \begin{gather*} b (c) = | c- a| \end{gather*} the dispersive action profile $f(c | a; \ebar)$ of the constant $0$-phase solution $v(x, t ; \ebar) = a$ with $s_0^{\uparrow}=a$. Replacing (\ref{ClassicalActionsRESULT}) derived in Steps~1--5 into the regular Bohr--Sommerfeld conditions (\ref{RegularBSConditionsGeneral}) gives \begin{gather*} 2 \pi \ebar \big| s_i^{\uparrow} - s_i^{\downarrow} \big| = 2 \pi \hbar N_i' \end{gather*} for $\ebar >0$, $\hbar >0$, and $N_i' \in \mathbb{Z}_+$ which is exactly formula (\ref{BOMultiPhaseBohrSommerfeldConditions}), the central claim of Part I.

\subsection{Step 7: Classical energy levels of multi-phase solutions}
In this step we confirm that the formula $\sum\limits_{i=0}^{n} \big(s_i^{\uparrow}\big)^3 - \sum\limits_{i=1}^n \big(s_i^{\downarrow}\big)^3$ in the statement of Part~II defines the classical energy levels of multi-phase solutions. From the Definition \ref{MultiPhaseProfileDEF} of the multi-phase profile $f(c | \vec{s})$ one calculates \begin{gather} \label{OhYesTheEnergy}\sum_{i=0}^n \big(s_i^{\uparrow}\big)^3 - \sum_{i=1}^n \big(s_i^{\downarrow}\big)^3 = \int_{- \infty}^{+\infty} c^3 \tfrac{1}{2} f ''(c | \vec{s}){\rm d}s. \end{gather} By Theorem~\ref{FiniteGapResult} and Proposition~\ref{Kerov337}, (\ref{OhYesTheEnergy}) is the Hamiltonian $O_3 (\ebar)|_v$ at $v=v^{\vec{s}, \vec{\chi}}(\cdot ; \ebar)$.

\subsection{Step 8: Quantum energy levels of Jack functions} In this step we derive a formula for the quantum energy levels of the quantum stationary states (Jack functions). By Proposition \ref{QuantumBOHamiltonianVIALAX}, the quantum periodic Benjamin--Ono Hamiltonian is
\begin{gather} \label{HereItIs} \widehat{O}_3(\ebar, \hbar) = 3\widehat{T}_3^{\uparrow}(\ebar, \hbar)- 3a \widehat{T}_2^{\uparrow}(\ebar, \hbar) + a^3 . \end{gather}
In (\ref{HereItIs}), $\widehat{O}_3(\ebar, \hbar)$ is expressed through the members $\widehat{T}_{\ell}^{\uparrow} (\ebar, \hbar)$ of the quantum Nazarov--Sklyanin hierarchy (\ref{QuantumNSHierarchy}). By Theorem~\ref{QuantumNSSPECTRUMtheorem} the eigenvalue of the quantum Hamiltonian $\widehat{O}_3(\ebar, \hbar)$ at the Jack function $P_{\lambda, a}(V | \ebar, \hbar)$ indexed by a partition $\lambda$ is \begin{gather} \label{GreatFormula} \widehat{O}_3(\ebar, \hbar) \Big |_{P_{\lambda, a}(\cdot | \ebar, \hbar)} = \int_{- \infty}^{+\infty} c^3 \tfrac{1}{2} f_{\lambda} '' (c - a | \varepsilon_2, \varepsilon_1){\rm d}c, \end{gather} where $f_{\lambda}(c -a | \varepsilon_2, \varepsilon_1)$ is the anisotropic partition profile of anisotropy $(\varepsilon_2, \varepsilon_1)$ centered at $a$ in Lemma~\ref{AnisotropicLemma}. Formula~(\ref{GreatFormula}) is the coefficient of $u^{-4}$ in the logarithmic derivative of~(\ref{NSMegaFormula}).

 \subsection{Step 9: Exact Bohr--Sommerfeld conditions} We now prove Part II of Theorem~\ref{MAINTHEOREM}. By formula (\ref{GreatFormula}), the exact spectrum of the quantum periodic Benjamin--Ono equation in $\overline{F}(a)$ is indexed by partitions $\lambda$ with quantum energy levels $\int_{-\infty}^{+\infty} c^3 \tfrac{1}{2} f_{\lambda} '' ( c -a|\varepsilon_2, \varepsilon_1){\rm d}c$. By formula (\ref{OhYesTheEnergy}), this quantum spectrum coincides with the classical energy levels of the multi-phase solutions whose multi-phase profiles $f( c| \vec{s})$ have $a = \sum\limits_{i=0}^n s_i^{\uparrow} - \sum\limits_{i=1}^n s_i^{\downarrow}$ and band and gap lengths
 \begin{gather*} 
 \big|s_i^{\downarrow} - s_{i-1}^{\uparrow} \big| = \varepsilon_1 N_i, \qquad
 \big|s_i^{\uparrow} - s_{i}^{\downarrow} \big| = - \varepsilon_2 N_i' . \end{gather*}
for $N_i, N_i' \in \mathbb{Z}_+$. These are the spatial periodicity conditions (\ref{BOMultiPhasePeriodicConditions}) and regular Bohr--Sommer\-feld conditions (\ref{BOMultiPhaseBohrSommerfeldConditions}) after the renormalization~(\ref{Renormalization}) of Abanov--Wiegmann \cite{AbWi1}. $\square$

\subsection*{Acknowledgments} The author would like to thank Chris Beasley, Percy Deift, Sam Johnson, Igor Krichever, Ryan Mickler, and Jonathan Weitsman for many helpful discussions. This work was supported by the Andrei Zelevinsky Research Instructorship at Northeastern University and also by the National Science Foundation RTG in Algebraic Geometry and Representation Theory under grant DMS-1645877.

\pdfbookmark[1]{References}{ref}
\LastPageEnding

\end{document}